\documentclass[12pt]{article}
\usepackage{graphicx}
\usepackage{natbib} 
\usepackage{url} 

\usepackage{subcaption}
\usepackage{hyperref}
\usepackage{algorithm}
\usepackage{algpseudocode}
\usepackage{amssymb}
\usepackage{amsmath}
\usepackage{amsthm}
\usepackage{xcolor}
\usepackage{booktabs}

\newcommand{\tajo}[1]{{\textcolor{blue}{tajo:#1}}}

\DeclareMathOperator*{\argmin}{argmin}
\DeclareMathOperator*{\argmax}{argmax}
\DeclareMathOperator*{\argsort}{argsort}
\newtheorem{theorem}{Theorem}

\newcommand{\blind}{1}

\begin{document}

\def\spacingset#1{\renewcommand{\baselinestretch}%
{#1}\small\normalsize} \spacingset{1}


\if1\blind
{
  \title{\bf A Fast Iterative Robust Principal Component Analysis Method}
  \author{Timbwaoga Aime Judicael Ouermi\thanks{
    The authors gratefully acknowledge \textit{This work was partially supported by the Intel OneAPI CoE, the Intel Graphics and Visualization Institutes of XeLLENCE, and the DOE Ab-initio Visualization for Innovative Science (AIVIS) grant 2428225.}}\hspace{.2cm}\\
    Jixian Li \\
    and \\
    Chris R. Johnson\\
    Scientific Computing and Imaging Institute, University of Utah}
  \maketitle
} \fi

\if1\blind
{
  \maketitle
  \medskip
} \fi

\bigskip
\begin{abstract}
Principal Component Analysis (PCA) is widely used for dimensionality reduction and data analysis. However, PCA results are adversely affected by outliers often observed in real-world data. Existing robust PCA methods are often computationally expensive or exhibit limited robustness.
In this work, we introduce a Fast Iterative Robust (FIR) PCA method by efficiently estimating the inliers center location and covariance. Our approach leverages Incremental PCA (IPCA) to iteratively construct a subset of data points that ensures improved location and covariance estimation that effectively mitigates the influence of outliers on PCA projection. We demonstrate that our method achieves competitive accuracy and performance compared to existing robust location and covariance methods while offering improved robustness to outlier contamination. We utilize simulated and real-world datasets to evaluate and demonstrate the efficacy of our approach in identifying and preserving underlying data structures in the presence of contamination. 
\end{abstract}

\noindent%
\vfill

\newpage
\spacingset{2} 
\section{Introduction}
\label{sec:introduction}
Principal component analysis (PCA)~\citep{Jolliffe2002pca} is a widely used technique for dimension reduction and data analysis. It provides a powerful framework for simplifying complex high-dimensional data, and uncovering major trends and patterns within the data. Despite its widespread application across many fields including cybersecurity~\citep{Priyanga2020,Parizad2022}, Machine Learning ~\citep{swathi2020overview,bharadiya2023tutorial}, and bio-science~\citep{Ma2011, Hubert_2004}, PCA is highly sensitive to outliers. Real-world datasets are often contaminated by outliers, which can distort statistical measures such as multivariate location, covariance, and the principal component scores derived from classical PCA (CPCA). This sensitivity can result in inaccurate representations and interpretations of the underlying structure of the data.

To mitigate the impact of outliers on PCA, several methods estimate the center \textit{location} and \textit{covariance} of inlier points\citep{Zhang02012024,Zhang_2024,hubert_fast_2002,hubert_robpca_2005,Hubert01072012}. The quality of robust PCA depends on how well the location and covariance represent the data characteristics.

~\cite{Rousseeuw01121984} introduced an approach based on finding the minimum covariance determinant (MCD). The MCD method consists of finding a subset of the data of user-specific size for which the covariance matrix has the smallest determinant, thereby ensuring that the influence of outliers is minimized. \cite{hubert_fast_2002,hubert_robpca_2005} used the accelerated and improved FastMCD~\citep{rousseeuw_fast_1999} to construct a robust PCA (RobPCA).  The FastMCD is improved to a deterministic MCD (detMCD) in ~\cite{Hubert01072012}, and later parallelized in ~\cite{DEKETELAERE2020103957} to improve its computational performance. A regularized and weighted regularized MCD is presented in ~\cite{boudt2020minimum} and ~\cite{kalina2022minimum}, respectively. ~\cite{schreurs2021outlier} developed a kernel-regularized MCD for non-elliptical data. These MCD-based methods require computationally expensive concentration steps.

In the context of functional data ~\cite{locantore1999robust} proposed a robust PCA method where they project data onto the unit sphere center at the location estimated using the minimum Euclidean norm to reduce the effect of the outliers.
~\cite{Maronna01112002} introduces the orthogonalized Gnanadesikan-Kettenring (OGK) estimator that utilizes a modified version of the robust covariance proposed by Gnanadesikan and Kettenring ~\cite{Gnanadesikan1972} to provide an estimation of location and covariance.
Recently, ~\cite{Zhang02012024} proposed a fast depth based (FDB) approach for estimating the location and covariance matrix for robust PCA. This approach calculates the projection, or $L^{2}$ depths(~\cite{Zuo2000}), and selects a subset of the data with the largest depth values to construct the location and covariance estimate. The methods mentioned above show limited robustness when applied to datasets with a large percentage of outliers near the main distribution and clustered around a single point, as observed in ~\cite{Hubert01072012,Zhang02012024}

The robustification of PCA to outliers depends on the quality of the location and covariance estimates. In addition, the computational cost is dominated by the search for the sampled subset used for calculating both location and covariance. To address the performance and robustness limitations, we introduce a novel \textbf{F}ast \textbf{I}terative \textbf{R}obust (FIR) location and covariance estimator that is robust to outliers. The corresponding robust PCA (FIR-PCA) constructed from this estimator maintains the robustness against outlier contamination, thereby ensuring reliable dimensionality reduction and enhanced robustness in the presence of anomalous data points. 

The key contributions of the work include: (1) A novel and fast method for location, covariance, and PCA projection estimation that is robust to outlier contamination. The proposed method leverages projection depth~\citep{Zuo2000} and incremental PCA (IPCA)~\citep{ross2008incremental} to build a fast approximation algorithm that minimizes the influence of outliers. (2) A theoretical analysis of orthogonal equivariance and permutation invariance for the FIR method. (3) A detailed analysis and evaluation of the novel FIR and FIR-PCA methods applied to several simulated and real-world examples.


%
\section{Methods}
\label{sec:methods}
Our method aims to efficiently provide accurate location and covariance estimations that enables the construction of a PCA framework that is robust against various types of outlier-contaminated datasets. We utilize projection depths~\citep{Zuo2000} and incremental PCA (IPCA)~\citep{ross2008incremental} to build the FIR location and covariance estimation. We can then obtain a robust PCA using the estimated location and covariance. \autoref{fig:overview} demonstrate the FIR procedure described in \autoref{algo:fir-algo}.

\subsection{Incremental PCA}
\label{subsec:IcrementalPCA}
\begin{figure}
    \centering
    \includegraphics[width=\linewidth]{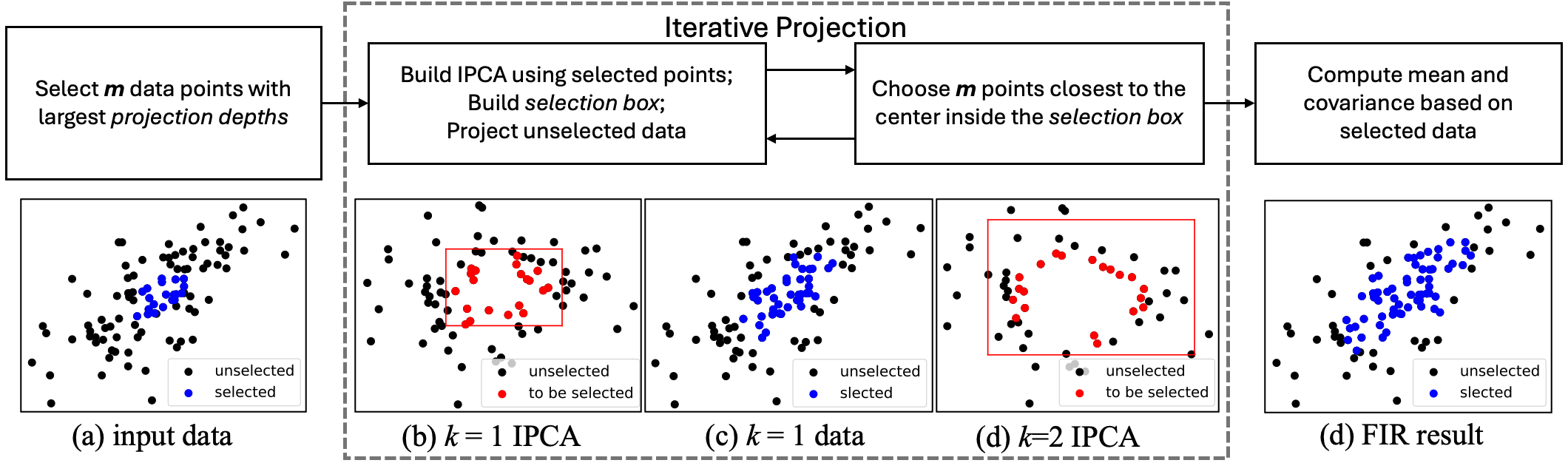}
    \caption{Overview of FIR pipeline. }
    \label{fig:overview}
\end{figure}
Incremental PCA (IPCA) is low dimensional subspace learning designed to process data incrementally~\citep{ross2008incremental,DIAZCHITO2018219,Weng2003,sarwar2002incremental}. Instead of requiring the entire dataset to be loaded into memory, IPCA updates the singular value decomposition (SVD) as new data arrive. This memory-efficient benefit makes the IPCA particularly attractive for large datasets and streaming data ~\citep{Sun2011,Yin2015,YUE2022127538,Zhao2006}. Given a new batch of data, the SVD update in IPCA consist of adjusting the data mean, singular values, and singular vectors to reflect the influence of the newly added data.

FIR iteratively adds points to the selected subset of data as candidate inliers to get estimations of center locations and covariance. IPCA fits this approach seamlessly by allowing incremental insertion of data. FIR uses the IPCA projection to construct a Mahalanobis-like distance characterizing the outlyingness of data to determine the candidate inliers efficiently. In our implementation, we use the IPCA developed by ~\cite{ross2008incremental}, which is available as a part of the \textit{scikit-learn}\footnote{\url{https://scikit-learn.org/stable/api/sklearn.html}} library.

\subsection{Fast Iterative Robust Location and Covariance Estimation}
\label{subsec:robust-localtion-covariance}

\autoref{algo:fir-algo} describes our FIR location and covariance estimation. Let $\mathbf{Z}$ be a matrix of dimension $n \times p$ representing $n$ $p-$dimensional points. The FIR approach starts by choosing $m$ initial indices, which correspond to the points with the largest projection depths. Let $\mathbf{u}$ be a unit vector in $\mathbb{R}^p$, then $\mathbf{u}^T\mathbf{z}$ for a sample $\mathbf{z}$ is a one-dimensional projection of $\mathbf{z}$. The projection depth of a data point $\mathbf{z}$ with respect to a distribution $F$ is:
\begin{equation}\label{eq:projection-depth}
 D_{proj}(\mathbf{z}; F) = \bigg( 1 +\sup_{\|\mathbf{u}^{T}\|=1} \frac{|\mathbf{u}^{T} \mathbf{z} - \textrm{med}(\mathbf{u}^{T}\mathbf{Z})|}{\textrm{mad}(\mathbf{u}^{T}\mathbf{Z})}\bigg),    
\end{equation}
where $F$ is the distribution of $\mathbf{Z}$. $\textrm{med}$ is the univariate median and $\textrm{mad}$ is the univariate median absolute deviation $\textrm{mad}(\mathbf{Y}) = \textrm{med}(\mathbf{Y}-\textrm{med}(\mathbf{Y}))$. In practice, $\mathbf{u}$ is randomly sampled from a $p-$unit sphere.

The remaining points are iteratively chosen from the unselected set of points. At each iteration step the number of selected and unselected points increases and decreases by a factor $m$. The criteria for selecting the $m$ points at the $(k+1)^{th}$ iteration is dependent on the distance of each projected points scaled by the IPCA singular values squared and a selection box constructed based on the projected points onto the two leading singular vectors. The scaled distance in IPCA space is defined as follows: 
\begin{equation}\label{eq:bb-distance}
    d_{i}^{(k)}  = \sum_{j=1}^{\hat{p}} \bigg(\frac{\mathbf{\hat{Z}}_{i,j}^{(k)}}{ s^{(k)}_{j}}\bigg)^2, 
\end{equation}
where $[s^{(k)}_{1}, s^{(k)}_{1}, \cdots,s^{(k)}_{\hat{p}}]$ are the non-zero singular values of the IPCA projection at iteration $k$. The selection box is calculated by expanding the bounding box of the points added at the $k^{th}$ iteration by a factor of 0.5 in all directions.

\begin{equation}\label{eq:bounding-box}
    \mathbf{b}_{j} = \big[\big(\max\big\{\mathbf{\hat{Z}}_{\ell,j}^{(k)}\big\}_{\ell \in H^{(k)}} + \delta_{j}\big), \big(\min \big\{\mathbf{\hat{Z}}_{\ell,j}\big\}_{\ell \in H^{(k)}} - \delta_{j}^{(k)}\big)\big]\quad j=1,2  \\
\end{equation}
where $\max\big\{\mathbf{\hat{Z}}_{\ell,j}^{(k)}\big\}_{\ell \in H^{(k)}}$ and $\min \big\{\mathbf{\hat{Z}}_{\ell,j}\big\}_{\ell \in H^{(k)}}$ are the upper and lower bounds for the two leading singular vector indicated by $j=1,2$ and $0.5 \big|\max\big\{\mathbf{\hat{Z}}_{\ell,j}^{(k)}\big\}_{\ell \in H^{(k)}} -\min \big\{\mathbf{\hat{Z}}_{\ell,j}\big\}_{\ell \in H^{(k)}}\big|$.

\autoref{fig:overview} provides an overview of our algorithm.
\autoref{fig:overview}a shows the input data, with the first selected $m$ points highlighted in blue, corresponding to those with the highest projection depth values. The FIR iterative process, as described in step (3) of \autoref{algo:fir-algo}, is depicted by the dashed gray box and the two arrows connecting the data and IPCA spaces. The initial $m$ points with highest projection depth values are used to build the first IPCA update and project the unselected points, as shown in \autoref{fig:overview}b From these projections, the $m$ points with the smallest distance within the red selection box (marked in red) are chosen, increasing the total number of selected points to $2m$. The IPCA projection is then updated to reflect the newly added points shown in \autoref{fig:overview}c The remaining unselected points are mapped to updated IPCA space to identify the next $m$ points, as shown in \autoref{fig:overview}d This iterative process continues until $0.5(p+n+1) \leq h=\alpha n < n$ points are selected shown in \autoref{fig:overview}d used for calculating the location and covariance. $\alpha$ is a user-specified parameter indicating the percentage of the data to be used for estimating location and covariance. Below, we outline the algorithm for the FIR location and covariance estimation where $\mathbf{Z}$ is the input data, $m$  is the batch size, and $\alpha$ represents the percentage of data used to construct the selected index set $H$, location $\mathbf{\mu}$, and covariance $\mathbf{\Sigma}$.

\begin{algorithm}[H]
    \caption{Fast Iterative Robust(FIR) Location and Covariance Estimation}
    \label{algo:fir-algo}
    \begin{algorithmic}
        \State \textbf{Input}: $\mathbf{Z}$, $m$, $\alpha$. \textbf{Output}: $\mathbf{\mu}$, $\mathbf{\Sigma}$, $H$
        \begin{enumerate}
        \item Calculate and sort  projection depth values using \autoref{eq:projection-depth}
        \itemsep0em 
        \item Choose $m$ indices with largest depth values to form $H^{(1)}= \{\ell_{1}, \ell_{2}, \cdots, \ell_{m}\}$
        \item For $k=1$ to $\lfloor h/m -1\rfloor$
            \begin{enumerate}
                \item Update IPCA projection and selection box using \autoref{eq:bounding-box}
                \item Project the unselected points to the IPCA space
                \item Calculate the scaled distance of the projected points using \autoref{eq:bb-distance}
                \item Choose the $m$ points with the smallest distance inside the selection box and update $H^{(k+1)} = H^{(k)} \cup \{\ell_{mk+1}, \ell_{mk+2}, \cdots, \ell_{mk+m}\}$
            \end{enumerate}
        \item Compute mean $\mathbf{\mu}$ and covariance of $\mathbf{\Sigma}$ using index subset $H$
        \end{enumerate}
    \end{algorithmic}
\end{algorithm}

\subsection{Fast Iterative Robust PCA}
\label{subsec:bounding-box-depth-robust-pca}
We can construct the FIR-PCA leveraging the FIR location and covariance approximations from \autoref{subsec:robust-localtion-covariance} to obtain the projection. Let $\mathbf{X}$ be the input matrix data with $n$ sample points and $p$ features. We begin by performing a singular value decomposition on the centered and normalized data according to
    $(\mathbf{X} - \mathbf{1} \otimes\mathbf{\mu}^{T}_{0})/\sqrt{n} = \mathbf{U} \mathbf{S} \mathbf{V}^{T}$, 
where $\mathbf{1}$ is column vector of ones and $\mathbf{\mu}_{0}$ is the classical mean of $\mathbf{X}$. The symbol $\otimes$ denotes the Kronecker product.
The original data is then projected onto the subspace spanned by the $r_{0}$ column vectors of $\mathbf{V}$ with nonzero singular values as follows:
\begin{equation}\label{eq:projection_onto_V}
    \mathbf{Z} = \mathbf{X} \mathbf{V}_{:,r_{0}},
\end{equation}
where $\mathbf{Z}$ is of dimension $n\times r_{0}$, and $\mathbf{V}_{:,r_{0}}$ denote the first $r_{0}$ columns of $\mathbf{V}$. The step \autoref{eq:projection_onto_V} is a data transformation without loss of information. \autoref{algo:fir-algo} is applied to $\mathbf{Z}$ for estimating the robust location $\mathbf{\mu}$ and covariance $\mathbf{\Sigma}$ used to calculate the robust scores
\begin{equation}\label{eq:T-scores}
    \mathbf{T} = \big(\mathbf{Z} - \mathbf{\mu}\big)\mathbf{P}_{:,r_{1}},  
\end{equation}
where $\mathbf{\Sigma} = \mathbf{P} \mathbf{L} \mathbf{P}^{T}$ and $\mathbf{P}_{:,r_{1}}$ represents the first $r_{1}$ column vectors of $\mathbf{P}$ with nonzero eigenvalues in $\mathbf{L}$. The robust principal vectors are chosen to be $\mathbf{V}_{:, r_{1}}$ and the variance is 
\begin{equation}\label{eq:sigma-BBD}
    \mathbf{\sigma} = diag \big( \mathbf{L} )^2.
\end{equation}
The robust PCA mean is obtained from 
\begin{equation}\label{eq:mu1}
    \mu_{1} = \mu_{0} + \mu \mathbf{V}^{T}_{r_{1}, :}.
\end{equation}
Identifying outliers is critical to isolating potential anomalies or features that deviate from the underlying structure of the data. These outliers are identified using the robust orthogonal and score distances of the data points, defined as follows:
\begin{equation}\label{eq:score-dist}
    sd_{i} = \sqrt{\sum_{j=1}^{r_{1}} \frac{t^{2}_{i,j}}{l_{j}}}
\end{equation}
\begin{equation}\label{eq:orthogonal-dist}
    od_{i} = ||\mathbf{x}_{i}- \mathbf{\mu}_{1} - \mathbf{P}_{r_{0}, k}\mathbf{t}_{i}^{T}||.
\end{equation}
\autoref{algo:fir-pca} outlines the steps for the FIR-PCA method where $\mathbf{X}$ is the input data, $m$ is the batch size, and $\alpha$ represents the percentage of data used to estimate the location and covariance in \autoref{algo:fir-algo}. The outputs consist of $\mathbf{T}$ (scores), $\mathbf{V}$ (principal components), $\sigma$ (variances), $\mu$ (mean), $sd$ (score distances), and $od$ (orthogonal distances).
\begin{algorithm}[H]
    \caption{ Fast Iterative Robust PCA}
    \label{algo:fir-pca}
    \begin{algorithmic}
        \State \textbf{Input:} $\mathbf{X}$, $m$, $\alpha$ \textbf{Output:} $\mathbf{T}$, $\mathbf{V}$, $\sigma$, $\mu$, $sd$, $od$
         \begin{enumerate}
            \item Compute $\mathbf{Z}$ from \autoref{eq:projection_onto_V}
            \item Compute $H,\mathbf{\mu}$, and $\mathbf{\Sigma}$ using \autoref{algo:fir-algo}
            \item Calculate $\mathbf{T}$, $\mathbf{V}_{:, r_{1}}$, $\mathbf{\sigma}$, and $\mathbf{\mu_{1}}$ using \autoref{eq:T-scores}, \autoref{eq:sigma-BBD}, and \autoref{eq:mu1} 
            \item Calculate score and orthogonal distance for each projected point from \autoref{eq:score-dist} and \autoref{eq:orthogonal-dist}
            \item Separate inliers and outliers indices based on the computed distances 
         \end{enumerate}
    \end{algorithmic}
\end{algorithm}

\section{Properties of the FIR Method}
\label{sec:propreties}

\subsection{Invariant Properties}
\label{subsec:invariant-prperties}
\textbf{\textit{Orthogonal equivariance}} refers to a property where the method results are invariant under translation and orthogonal transformation. More precisely, for any nonsingular orthogonal matrix $\mathbf{A}$ with dimensions $p \times p$ and vector $\mathbf{v}$ with dimension $1\times p$ the method is \textit{orthogonal equivariant} if 
\begin{equation}\label{eq:equivariance-mu}
    \mathbf{\mu \big( Z A + 1_{n} v^{T}\big) = \mu\big( Z\big)A + v},
\end{equation}
\begin{equation}\label{eq:equivariance-sigma}
    \mathbf{\Sigma \big( ZA +  1_{n} v^{T}\big) = A^{T}\Sigma\big( Z\big)A}
\end{equation}
\begin{theorem}\label{theo:theo1}
    Let $\mathbf{Z} \in \mathbb{R}^{n \times p}$ be a nonsingular data matrix with $n$ observation and $p$ features. If $\mathbf{A}$ is an orthogonal transform, $\mathbf{v}$ a vector, and $m$ the batch size, the scaled distances $d_{i}^{(k)}  = \sum_{j=1}^{\hat{p}} \bigg(\frac{\mathbf{\hat{Z}}_{i,j}^{(k)}}{ s^{(k)}_{j}}\bigg)^2$ from \autoref{eq:bb-distance} are invariant for $\mathbf{Z}$ and $\mathbf{Z A} + \mathbf{1}_{n} \mathbf{v}^{T}$ and the FIR method is \textit{orthogonal equivariant}.
\end{theorem}
\begin{proof}
    The initial $m$ indices are obtained from the projection depth defined in \autoref{eq:projection-depth}.The projection depth is \textit{affine equivariant}~\cite{Zuo2000} and consequently it is also \textit{orthogonal equivariant}. The orthogonal relationship for the first $m$ selected point can be expressed as follows: 
    \begin{equation}\label{eq:equivariance-mu}
        \mathbf{\mu} \big( \mathbf{Z}_{H^{(1)}} \mathbf{A} + \mathbf{1}_{m} \mathbf{v}^{T}\big) = \mathbf{\mu}\big( \mathbf{Z}_{H^{(1)}}\big)\mathbf{A + v},
    \end{equation}
    \begin{equation}\label{eq:equivariance-sigma}
        \mathbf{\Sigma} \big( \mathbf{Z}_{H^{(1)}}\mathbf{A} +  \mathbf{1}_{m} \mathbf{v}^{T}\big) = \mathbf{A}^{T}\mathbf{\Sigma}\big( \mathbf{Z}_{H^{(1)}}\big)\mathbf{A}, 
    \end{equation}
    where $H^{(1)}$ denotes the initial subset of indices selected using the projection depth.
    
    Noting that the IPCA projection scores and singular values constructed from the index set $H^{(1)}$ is invariant under rotation and translation, we write that
    \begin{equation}\label{eq:ipca-score-H1}
        f_{H^{(1)}}(\mathbf{Z}) = \tilde{f}_{H^{(1)}}(\mathbf{Z}  \mathbf{A} + \mathbf{1}_{n} \mathbf{v}^{T})= \hat{\mathbf{Z}}, \textrm{ and}
    \end{equation}
    \begin{equation}\label{eq:ipca-singular-values-H1}
        \mathbf{s}_{(\mathbf{Z}_{H^{(1)}})} =\mathbf{s}_{(\mathbf{Z}_{H^{(1)}}  \mathbf{A} + \mathbf{1}_{n} \mathbf{v}^{T})}= \mathbf{s},
    \end{equation}
    where the functions $f_{H^{(1)}}$ and  $\tilde{f}_{H^{(1)}}$ constructed from IPCA fit of $\mathbf{Z}_{H^{(1)}}$ and $\mathbf{Z}_{H^{(1)}}  \mathbf{A} + \mathbf{1}_{n} \mathbf{v}^{T}$, respectively, project the data onto the corresponding IPCA spaces. The term $\mathbf{s}_{\mathbf{Y}}$ denote the singular values of $\mathbf{Y}$. From the results in \autoref{eq:ipca-score-H1} and \autoref{eq:ipca-singular-values-H1}, it follows that the scaled distances $d_{i}^{(1)}  = \sum_{j=1}^{\hat{p}} \bigg(\frac{\mathbf{\hat{Z}}_{i,j}^{(1)}}{ s^{(1)}_{j}}\bigg)^2$, in \autoref{eq:bb-distance}, and the new indices added form $H^{(2)}$ are the identical for both $\mathbf{Z}$ and $\mathbf{Z A} + \mathbf{1}_{n} \mathbf{v}^{T}$. These results imply that the set of $m$ points with smallest scaled distance that lie in the bounding box defined in \autoref{eq:bounding-box} is also identical for both $\mathbf{Z}$ and $\mathbf{Z A} + \mathbf{1}_{n} \mathbf{v}^{T}$ and following \textit{orthogonal equivariance} relationship based on the $2m$ selected points holds
     \begin{equation}\label{eq:equivariance-mu-H2}
        \mathbf{\mu} \big( \mathbf{Z}_{H^{(2)}} \mathbf{A} + \mathbf{1}_{2m} \mathbf{v}^{T}\big) = \mathbf{\mu}\big( \mathbf{Z}_{H^{(2)}}\big)\mathbf{A + v},
    \end{equation}
    \begin{equation}\label{eq:equivariance-sigma-H2}
        \mathbf{\Sigma} \big( \mathbf{Z}_{H^{(2)}}\mathbf{A} +  \mathbf{1}_{2m} \mathbf{v}^{T}\big) = \mathbf{A}^{T}\mathbf{\Sigma}\big( \mathbf{Z}_{H^{(2)}}\big)\mathbf{A}, 
    \end{equation}
    where $H^{(2)}$ denote the subset of indices constructed by combining $H^{(1)}$ and newly added indices.
    By induction, the procedure from \autoref{eq:ipca-score-H1} to \autoref{eq:equivariance-sigma-H2} holds for $H^{(2)}$, and subsequent index sets $H^{(k)}$. The final index set $H$ is identical for both  $\mathbf{Z}$ and $\mathbf{Z A} + \mathbf{1}_{n} \mathbf{v}^{T}$ and the FIR method \textit{orthogonal equivariant}.
\end{proof}

\textbf{\textit{Permutation Invariant}}: The permutation invariant requires that applying a permutation matrix that shuffles the data doesn't change the results. Because the projection depth, the IPCA, and the scaled distances defined in \autoref{eq:bb-distance} are permutation invariant. Applying a permutation matrix to data doesn't change the index set selected at each iteration step in the FIR methods outlined in \autoref{algo:fir-algo}.

\subsection{Robustness}
\label{subsec:robusness}
Robustness to outliers is crucial for accurate statistical approximations, subsequent analysis, and interpretation of the data. To evaluate the robustness of FIR compared to DetMCD, and FBD, we consider \textbf{Cluster} and \textbf{Point} datasets with $40\%$ outliers described in \autoref{sec:results}. The example in \autoref{fig:roubustness_eval} is based on $n=1000$ observations and $p=10$ features. The scatter plots of the last two ($9$ and $10$) feature axes show the subset of points selected by DetMC, FDB, and FIR from left to right, with the outliers indicated in red. All three methods successfully select inliers in the cases with clustered outliers, shown in the top row of \autoref{fig:roubustness_eval}. In the examples with point outliers in the bottom row, the DetMCD and FDB methods include the outliers, whereas the FIR correctly excludes them from the selection of inliers. These results demonstrate that the FIR methods exhibit improved robustness to outlier contamination compared to FDB and DetMCD. 
\begin{figure}[H]
    \centering
    \begin{subfigure}[t]{0.32 \textwidth}
        \includegraphics[width=0.8\linewidth]{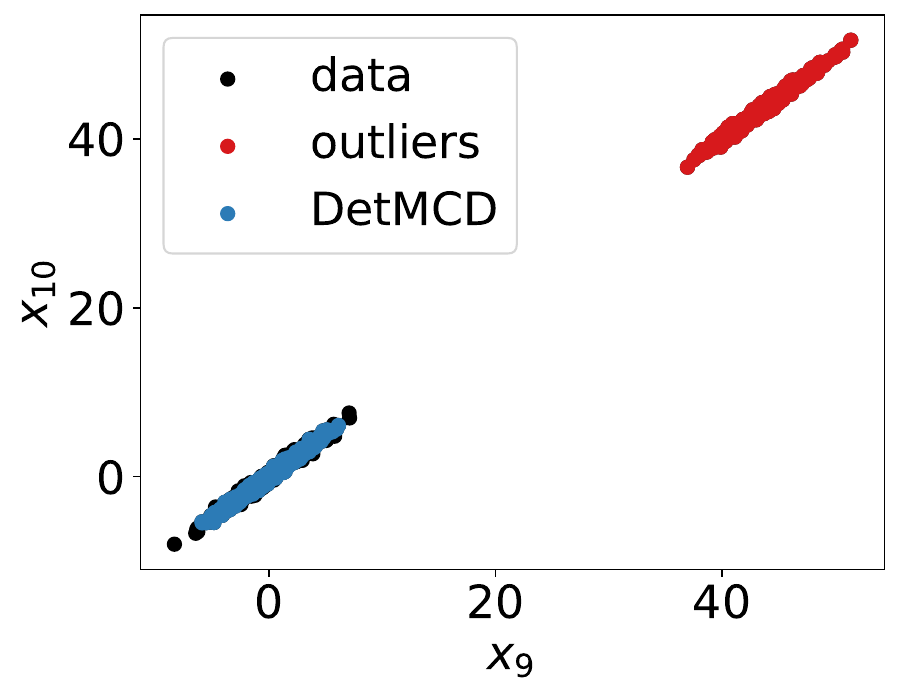}
        \label{subfig:DetMCD_cluster_outliers}
    \end{subfigure}
    \begin{subfigure}[t]{0.32 \textwidth}
        \includegraphics[width=0.8\linewidth]{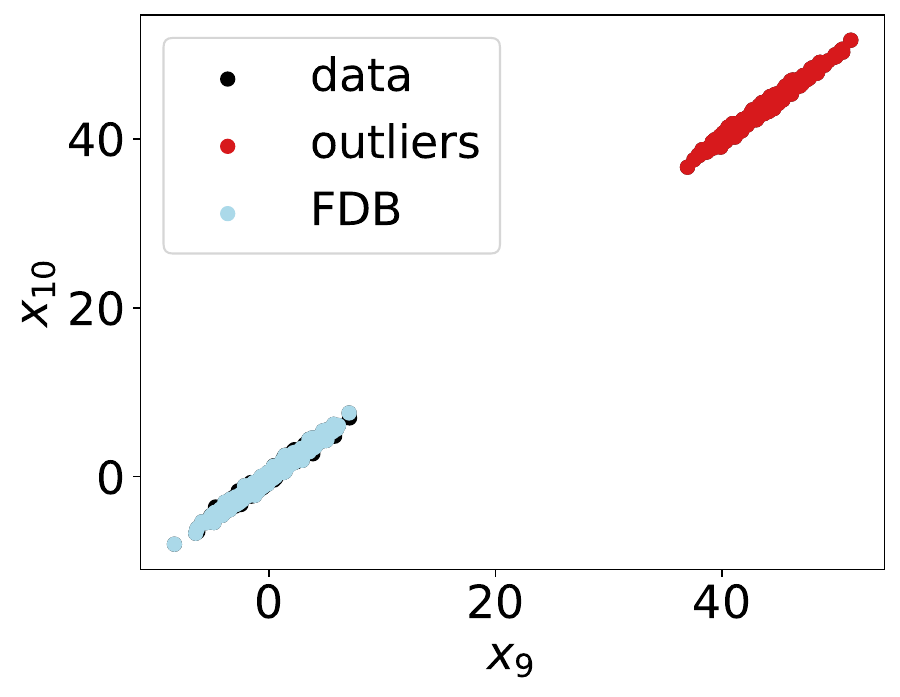}
        \label{subfig:FDB_cluster_outliers}
    \end{subfigure}
    \begin{subfigure}[t]{0.32 \textwidth}
        \includegraphics[width=0.8\linewidth]{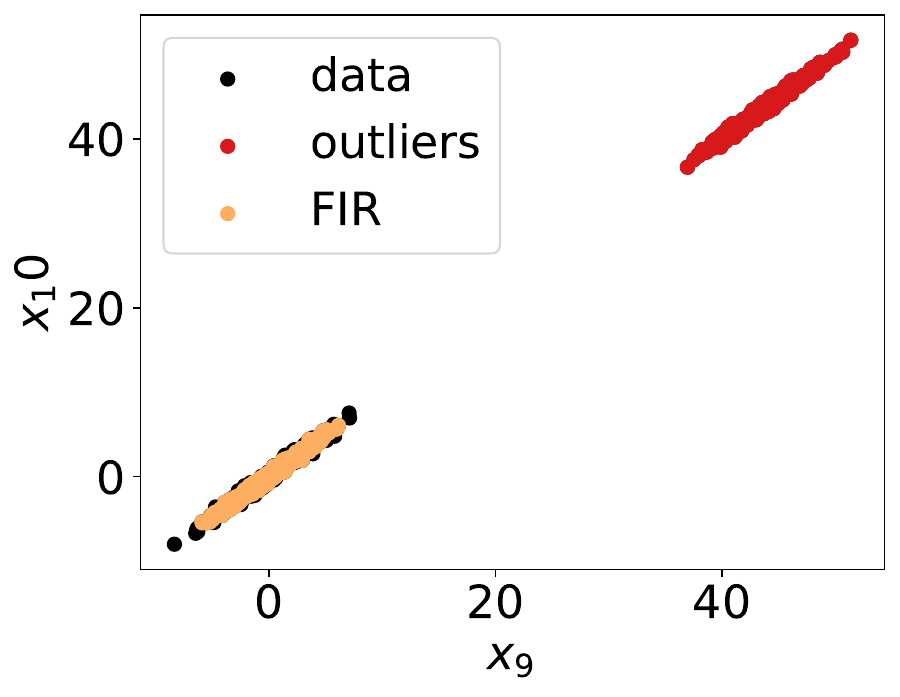}
        \label{subfig:FIR_cluster_outliers}
    \end{subfigure}
    \begin{subfigure}[t]{0.32 \textwidth}
        \includegraphics[width=0.8\linewidth]{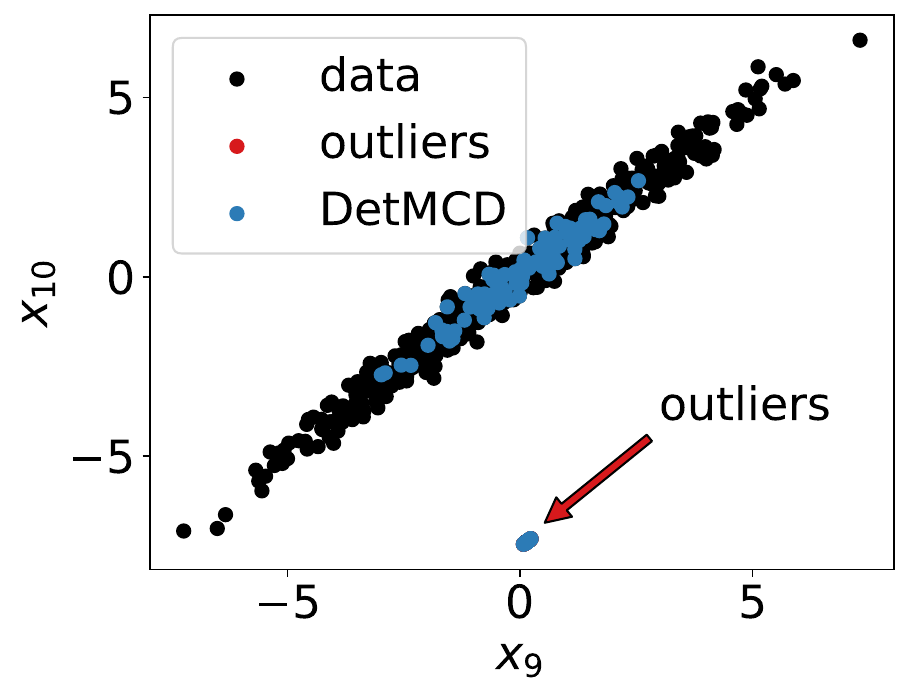}
        \caption{DetMCD $p=10$}
        \label{subfig:DetMCD_point_outliers}
    \end{subfigure}
    \begin{subfigure}[t]{0.32 \textwidth}
        \includegraphics[width=0.8\linewidth]{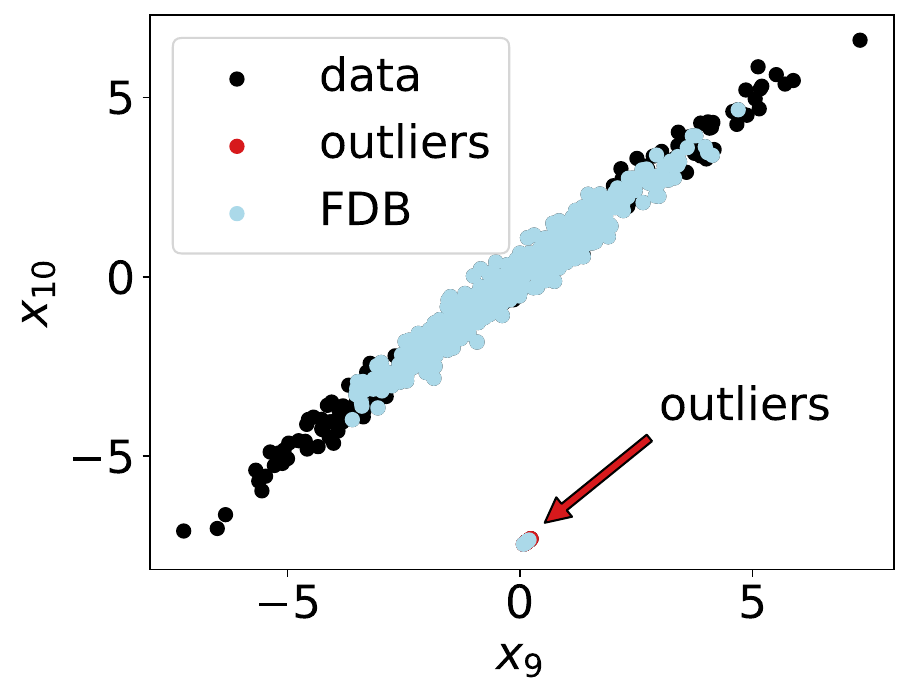}
        \caption{FDB $p=10$}
        \label{subfig:FDB_point_outliers}
    \end{subfigure}
    \begin{subfigure}[t]{0.32 \textwidth}
        \includegraphics[width=0.8\linewidth]{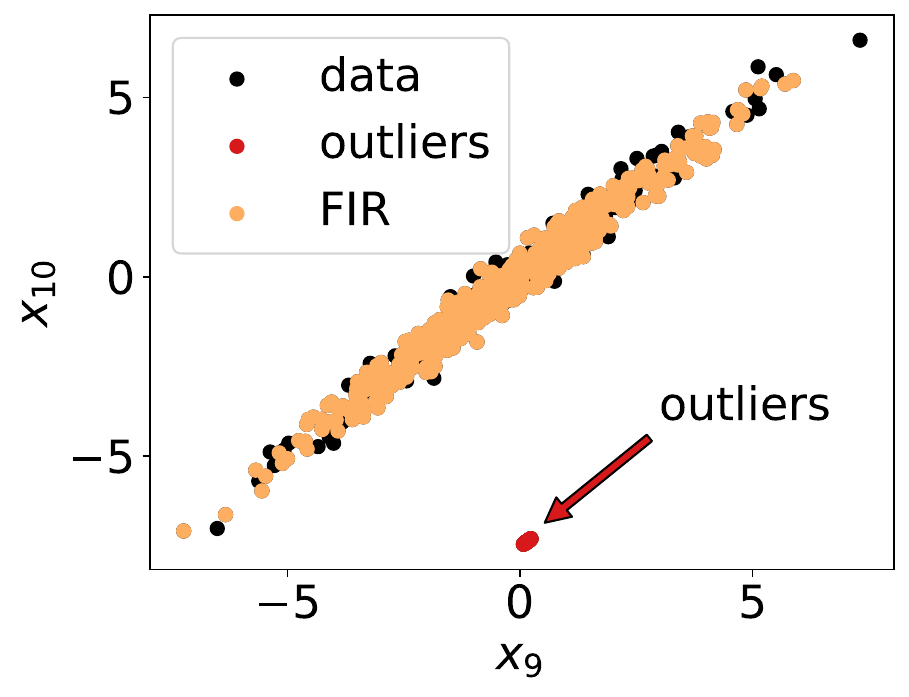}
        \caption{FIR $p=10$}
        \label{subfig:FIR_point_outliers}
    \end{subfigure}
    \caption{Selected points with DetMCD, FDB, and FIR. The example uses $n=1000$ sample points with $40\%$ outliers. The top and bottom two rows correspond to the cluster and point outliers, respectively, with $h=0.5n$. The horizontal  and vertical axis corresponds to the 9th and 10th feature dimension, respectively.} 
    \label{fig:roubustness_eval}
\end{figure}

\subsection{Computational Complexity Analysis}
\label{subsec:computational-complexity-analysis}
The computational complexity of \autoref{algo:fir-algo} depends on the calculation of the projection depth values used for selecting the initial $m$ points, the IPCA, and the calculation of the scaled distances. The resulting computational complexity is $O(\kappa^{2}mp + \kappa mp^{2} + \tau np)=O(\kappa np + np^{2} + \tau np)= O((\kappa +\tau)np + np^2)$, where $\tau np $ denote the projection depth estimation, $\kappa mp^{2}$ corresponds to the IPCA updates, and $\kappa^{2}mp$ to the scaled distances calculations. The term $\kappa$ denotes the number of batches. 

In comparison, The MCD-based methods have a computational complexity of $O\big(\Gamma(p^{3} + np^2)\big)$, where $\Gamma$ corresponds to the number of iterations of the C-step. The term $(p^{3} + np^2)$ is the cost of a single C-step, where $p^{3}$ corresponds to the inversion of the covariance matrix required for the Mahalanobis distance computation, and $np^2$ accounts for the Mahalanobis distance calculation for all the observations. The FBD method ~\citep{Zhang02012024} has a computational complexity of $O(\tau np)$, corresponding to the cost of estimating projection depths. The parameter $\tau$ corresponds to the number of sample directions used in the depth calculations. As we demonstrate in Section \autoref{subsec:performance}, the FDB and FIR methods are significantly faster than the DetMCD method, and the FIR method demonstrates improved robustness to outliers.

\section{Simulated Examples}
\label{sec:results}
We use simulated datasets to demonstrate the robustness and computational efficiency of the FIR method and compare to two existing methods: DetMCD and FBD.The datasets are obtained from $\mathbf{x}_{i} = \mathbf{Gy}_{i}$, where $\mathbf{y}_{i}$ are sampled from $\mathcal{N}_{p}(0, \mathbf{I})$, and $\mathbf{G}$ is a matrix of dimension $p\times p$ with ones on the diagonal, and $0.75$ on the off diagonals. We consider three types of outliers--\textbf{Cluster}, \textbf{\textbf{Radial}}, and \textbf{Point}-- which have been utilized in \cite{boudt2020minimum,Zhang02012024} to assess location and covariance approximation methods. The \textbf{Point} outliers are obtained from $\mathbf{y}_{i} \sim \mathcal{N}_{p}(r \sqrt{p} \mathbf{a}, 0.01^{2} \mathbf{I})$, where $\mathbf{a}$ is a unit vector orthogonal to $\mathbf{a}_{0}=(1, 1, \cdots, 1)^{T}$ and $r$ is a constant. 
The \textbf{Cluster} outliers are generated from $\mathbf{y}_{i} \sim \mathcal{N}_{p} (r p^{-1/4} \mathbf{a}_{0}, \mathbf{I})$. The \textbf{Radial} outliers are generated from $\mathbf{y}_{i} \sim \mathcal{N}_{p}(\mathbf{0}, 5\mathbf{I})$. 

We employ different evaluation metrics to compare the estimated quantities to the target solutions. The location error is Euclidean distance $e_{\mu} = ||\mathbf{\hat{\mu}} -\mathbf{\mu}||_2$ quantifies the deviation between the approximated $\mathbf{\mu}$ and target location $\mathbf{\hat{\mu}}$. The Kullback-Leibler (KL) divergence $e_{KL} = trace(\mathbf{\hat{\Sigma}} \mathbf{\Sigma}^{-1}) - log(det(\mathbf{\hat{\Sigma}} \mathbf{\Sigma}^{-1})) - p$ measures how much the estimation $\mathbf{\Sigma}$ is different than the target $\hat{\mathbf{\Sigma}}$. We calculate the mean square error (MSE), $e_{MSE} = \frac{1}{p^{2}} ||\mathbf{\hat{\Sigma}} - \mathbf{\Sigma}||_{F}$, where $F$ denotes the Forbenius norm. Lastly, we compare the different methods' computational runtime $t$ in seconds.

\subsection{Accuracy}
\label{subsec:accuracy}
Here, we evaluate the accuracy of location and covariance estimation from DetMCD, FDB, and FIR and their corresponding robust PCA. \autoref{tab:approx-errors-beta-10} and \autoref{tab:approx-errors-beta-40} report approximation errors of the DetMCD, FDB, and FIR for datasets \textbf{A} ($n=200, p=5$), \textbf{B} ($n=300, p=20$), \textbf{C} ($n=400, p=50$), and \textbf{D} ($n=1000, p=100$).
\textit{Clean} corresponds to datasets without outliers, and its corresponding results are replicated in \autoref{tab:approx-errors-beta-10} and \autoref{tab:approx-errors-beta-40}. The remaining results in \autoref{tab:approx-errors-beta-10} have $\%10$ outliers and $\alpha = 0.75$ while those in \autoref{tab:approx-errors-beta-40} correspond to the cases with $\%40$ outliers and $\alpha=0.5$.
%
The results are averaged of $1000$ samples of each dataset. The standard deviations are provided in parenthesis immediately after the mean. The DetMCD, FDB, and FIR provide reasonably good approximations for the scenarios with $\%10$ outliers for the \textbf{Cluster}, \textbf{Radial}, and \textbf{Point} outliers. In the examples with $\%40$ outliers, all three methods lead to comparable approximation for the datasets with 
\begin{table}[H]
    \caption{Estimation errors for DetMCD, FDB, and FIR methods with $10\%$ outliers.}
    \label{tab:approx-errors-beta-10}
    \centering
    \scriptsize
    \begin{tabular}{lccccccccl}
    \toprule
    & & \multicolumn{3}{c}{Clean} & \multicolumn{3}{c}{Cluster}  \\
     \cmidrule(r){3-5}\cmidrule(r){6-8}
    &  & DetMCD  &  FDB   & FIR & DetMCD  &  FDB   & FIR\\
    \toprule
A &  $e_{\mu}$      &  0.18 (0.05)&  0.10 (0.07)&  0.18 (0.06)&  0.25 (0.05)&  0.05 (0.08)&  0.13 (0.07)  \\
  &  $e_{\Sigma}$   &  0.44 (0.02)&  0.22 (0.02)&  0.25 (0.02)&  0.46 (0.02)&  0.35 (0.02)&  0.34 (0.02)  \\
  &  $e_{KL}$       &  4.07 (0.11)&  4.12 (0.10)&  4.34 (0.13)&  4.06 (0.11)&  4.42 (0.10)&  4.44 (0.11)  \\
B &  $e_{\mu}$      &  0.12 (0.22)&  0.34 (0.01)&  0.32 (0.10)&  0.41 (0.20)&  1.15 (0.10)&  1.25 (0.19)  \\
  &  $e_{\Sigma}$   &  0.49 (0.00)&  0.29 (0.02)&  0.45 (0.00)&  0.56 (0.01)&  0.35 (0.02)&  0.52 (0.02)  \\
  &  $e_{KL}$       &  24.05 (0.05)&  20.09 (0.39)&  24.90 (0.14)&  25.51 (0.33)&  21.18 (0.50)&  26.09 (0.43)  \\
C &  $e_{\mu}$      &  1.05 (0.22)&  0.70 (0.61)&  1.55 (0.36)&  1.06 (0.19)&  1.33 (0.03)&  1.11 (0.15)  \\
  &  $e_{\Sigma}$   &  0.52 (0.00)&  0.39 (0.00)&  0.51 (0.00)&  0.55 (0.01)&  0.45 (0.01)&  0.55 (0.00)  \\
  &  $e_{KL}$       &  69.52 (0.07)&  59.84 (0.00)&  70.24 (0.05)&  70.61 (0.81)&  63.64 (0.58)&  71.45 (0.09)  \\
D &  $e_{\mu}$      &  1.58 (0.77)&  1.98 (0.56)&  0.50 (1.30)&  0.63 (1.22)&  4.49 (1.28)&  1.24 (1.21)  \\
  &  $e_{\Sigma}$   &  0.47 (0.02)&  0.27 (0.02)&  0.45 (0.02)&  0.47 (0.01)&  0.34 (0.02)&  0.46 (0.01)  \\
  &  $e_{KL}$       &  132.15 (1.93)&  103.24 (2.62)&  130.63 (2.42)&  131.19 (1.10)&  112.54 (2.35)&  131.69 (1.29)  \\
& & \multicolumn{3}{c}{Radial} & \multicolumn{3}{c}{Point}  \\
     \cmidrule(r){3-5}\cmidrule(r){6-8}
A &  $e_{\mu}$ &  0.37 (0.06)&  0.21 (0.05)&  0.08 (0.08)&  0.28 (0.05)&  0.51 (0.08)&  0.08 (0.07)  \\
  &  $e_{\Sigma}$ &  0.48 (0.02)&  0.31 (0.02)&  0.34 (0.02)&  0.45 (0.02)&  0.26 (0.02)&  0.34 (0.02)  \\
  &  $e_{KL}$ &  4.16 (0.11)&  4.16 (0.09)&  4.43 (0.11)&  4.05 (0.11)&  3.81 (0.09)&  4.42 (0.10)  \\
B &  $e_{\mu}$ &  0.31 (0.18)&  0.51 (0.09)&  0.56 (0.16)&  0.27 (0.19)&  0.11 (0.12)&  0.67 (0.06)  \\
  &  $e_{\Sigma}$ &  0.56 (0.01)&  0.33 (0.03)&  0.49 (0.01)&  0.56 (0.01)&  0.31 (0.02)&  0.49 (0.00)  \\
  &  $e_{KL}$ &  25.42 (0.25)&  20.39 (0.72)&  25.31 (0.20)&  25.58 (0.29)&  20.13 (0.42)&  25.13 (0.08)  \\
C &  $e_{\mu}$ &  1.59 (0.31)&  0.40 (0.61)&  0.54 (0.33)&  1.09 (0.20)&  1.21 (0.44)&  0.65 (0.05)  \\
  &  $e_{\Sigma}$ &  0.55 (0.02)&  0.44 (0.01)&  0.52 (0.01)&  0.54 (0.01)&  0.39 (0.00)&  0.54 (0.00)  \\
  &  $e_{KL}$ &  70.87 (1.03)&  62.88 (0.71)&  69.85 (0.59)&  70.46 (0.88)&  59.96 (0.21)&  71.09 (0.16)  \\
D &  $e_{\mu}$ &  0.56 (1.27)&  0.44 (1.12)&  1.23 (0.92)&  0.51 (1.27)&  0.93 (1.51)&  1.09 (1.26)  \\
  &  $e_{\Sigma}$ &  0.46 (0.01)&  0.31 (0.02)&  0.46 (0.01)&  0.47 (0.02)&  0.27 (0.02)&  0.45 (0.01)  \\
  &  $e_{KL}$ &  129.37 (1.56)&  109.54 (2.21)&  131.12 (1.63)&  150.14 (69.64)&  104.32 (2.98)&  129.73 (1.37)  \\
\bottomrule
\end{tabular}
\end{table}
\begin{table}[H]
    \caption{Estimation errors for DetMCD, FDB, and FIR methods with $40\%$ outliers.}
    \label{tab:approx-errors-beta-40}
    \centering
    \scriptsize
    \begin{tabular}{lccccccccl}
    \toprule
    & & \multicolumn{3}{c}{Clean} & \multicolumn{3}{c}{Cluster}  \\
     \cmidrule(r){3-5}\cmidrule(r){6-8}
    &  & DetMCD  &  FDB   & FIR & DetMCD  &  FDB   & FIR\\
    \toprule
A &  $e_{\mu}$      &  0.18 (0.05)&  0.10 (0.07)& 0.18 (0.06)&     0.30 (0.05)&  0.36 (0.06)&  0.43 (0.08)  \\
  &  $e_{\Sigma}$   &  0.44 (0.02)&  0.22 (0.02)& 0.25 (0.02)&     0.47 (0.02)&  0.44 (0.02)&  0.34 (0.03)  \\
  &  $e_{KL}$       &  4.07 (0.11)&  4.12 (0.10)& 4.34 (0.13)&    4.12 (0.07)&  4.88 (0.07)&  4.58 (0.09)  \\
B &  $e_{\mu}$      &  0.12 (0.22)&  0.34 (0.01)& 0.32 (0.10)&     0.61 (0.12)&  0.53 (0.06)&  0.35 (0.04)  \\
  &  $e_{\Sigma}$   &  0.49 (0.00)&  0.29 (0.02)& 0.45 (0.00)&      0.57 (0.00)&  0.56 (0.00)&  0.46 (0.01)  \\
  &  $e_{KL}$       &  24.05 (0.05)&  20.09 (0.39)&24.90 (0.14)&  25.37 (0.06)&  26.01 (0.01)&  25.02 (0.23)  \\
C &  $e_{\mu}$      &  1.05 (0.22)&  0.70 (0.61)& 1.55 (0.36)&     1.36 (0.34)&  0.22 (0.37)&  1.59 (0.38)  \\
  &  $e_{\Sigma}$   &  0.52 (0.00)&  0.39 (0.00)& 0.51 (0.00)&       0.59 (0.00)&  0.63 (0.00)&  0.57 (0.01)  \\
  &  $e_{KL}$       &  69.52 (0.07)&  59.84 (0.00)&70.24 (0.05)&  75.12 (0.16)&  76.02 (0.13)&  75.40 (0.80)  \\
D &  $e_{\mu}$      &  1.58 (0.77)&  1.98 (0.56)& 0.50 (1.30)&     0.63 (1.89)&  1.15 (1.71)&  1.08 (1.73)  \\
  &  $e_{\Sigma}$   &  0.47 (0.02)&  0.27 (0.02)& 0.45 (0.02)&      0.50 (0.00)&  0.49 (0.01)&  0.48 (0.01)  \\
  &  $e_{KL}$       &  132.15 (1.93)&  103.24 (2.62)&130.63 (2.42)& 135.33 (0.43)&  134.41 (1.18)&  137.48 (0.49)  \\
 & & \multicolumn{3}{c}{Radial} & \multicolumn{3}{c}{Point}  \\
     \cmidrule(r){3-5}\cmidrule(r){6-8}
A &  $e_{\mu}$ &  0.31 (0.01)&  0.26 (0.03)&  0.31 (0.21)&  3.01 (0.10)&  1.37 (0.25)&  0.34 (0.10)  \\
  &  $e_{\Sigma}$ &  0.64 (0.02)&  0.29 (0.02)&  0.44 (0.00)&  0.33 (0.00)&  0.23 (0.07)&  0.33 (0.00)  \\
  &  $e_{KL}$ &  4.14 (0.07)&  4.17 (0.05)&  5.26 (0.07)&  34.20 (0.06)&  21.61 (7.92)&  4.52 (0.07)  \\
B &  $e_{\mu}$ &  0.69 (0.15)&  0.33 (0.02)&  0.35 (0.04)&  8.40 (0.05)&  1.40 (0.26)&  0.32 (0.14)  \\
  &  $e_{\Sigma}$ &  0.56 (0.00)&  0.38 (0.03)&  0.40 (0.01)&  0.06 (0.00)&  0.24 (0.02)&  0.40 (0.03)  \\
  &  $e_{KL}$ &  25.18 (0.00)&  21.90 (0.70)&  23.46 (0.24)&  140.62 (0.91)&  48.80 (3.71)&  23.61 (0.53)  \\
C &  $e_{\mu}$ &  0.26 (0.50)&  1.07 (0.00)&  0.53 (0.49)&  14.25 (0.00)&  1.88 (0.28)&  0.30 (0.25)  \\
  &  $e_{\Sigma}$ &  0.62 (0.00)&  0.47 (0.01)&  0.56 (0.02)&  0.04 (0.00)&  0.43 (0.02)&  0.49 (0.00)  \\
  &  $e_{KL}$ &  75.96 (0.09)&  66.24 (0.58)&  74.86 (1.70)&  401.35 (0.63)&  64.65 (5.02)&  69.75 (0.21)  \\
D &  $e_{\mu}$ &  0.68 (1.71)&  0.66 (1.66)&  0.44 (1.06)&  20.02 (0.00)&  10.94 (0.50)&  2.76 (0.29)  \\
  &  $e_{\Sigma}$ &  0.49 (0.00)&  0.38 (0.02)&  0.48 (0.00)&  0.01 (0.01)&  0.07 (0.01)&  0.44 (0.00)  \\
  &  $e_{KL}$ &  134.50 (0.56)&  121.50 (1.96)&  138.22 (0.20)&  711.00 (1.85)&  1222.15 (54.34)&  132.57 (0.22)  \\
\bottomrule
\end{tabular}
\end{table}
\textbf{Cluster} and \textbf{Radial} outliers. However, for the datasets with $\%40$ \textbf{Point} outliers the DetMCD and FDB have significantly larger location errors, and KL divergence values than the FIR approach.

\autoref{fig:approx-errors-various-p} visualizes the estimation errors from the three methods for various $p$ feature dimensions. The examples uses $n=1000$ observation, $10\%$ outlier contamination and $\alpha=0.75$. The first, second, and third rows show the location error $e_{\mu}$, covariance mean squared error $e_{\Sigma}$, and the KL divergence, respectively. These figures show comparable errors except for the \textbf{Point} outliers where the FIR method has smaller location errors and KL divergence values compared to the DetMCD and FDB methods. \autoref{fig:approx-error-noise-level} evaluates the effect of outlier contamination level for the different methods on the \textbf{Point} outlier example with $n=1000$, $p=5$, and $\alpha=0.5$. The results show comparable errors for noise levels less than $20\%$. As the noise level increases beyond $20\%$ the location error and KL divergence significantly increase. The results in \autoref{tab:approx-errors-beta-10} \autoref{tab:approx-errors-beta-40}, \autoref{fig:approx-errors-various-p}, and \autoref{fig:approx-error-noise-level} demonstrate that the FIR method achieves an accuracy comparable to DetMCD and FDB in the example with a small percentage of outliers while demonstrating greater robustness for examples with a larger percentage of outliers.

\autoref{fig:robust-PCA-comparison} shows the scores of the first two principal components for the classical PCA, DetMC-PCA, FDB-PCA, and FIR-PCA. The dataset used in \autoref{fig:robust-PCA-comparison} is generated by calculating the $\mathbf{U V}$ and injecting $10\%$ outliers by adding $20|\mathbf{y}|$ to randomly selected rows with $\mathbf{y}$ drawn from  $\mathcal{N}_{p}(\mathbf{I}, 1)$. Here, $\mathbf{U}$ and $\mathbf{V}$ are low rank matrices of size $n\times rank$ and $rank \times p$, respectively, with $n=300$ observation, $p=10$ features,and $rank=2$. The results in \autoref{fig:robust-PCA-comparison} demonstrate that the CPCA does not mitigate the influence of the outliers, leading to incorrect location and variances estimations, as indicated by the elongated line segments in \autoref{subfig:PCA}. In contrast, the robust PCA methods successfully identify the cluster with the correct corresponding variances.  

\begin{figure}[H]
    \centering
    \begin{subfigure}[t]{0.32 \textwidth}
    \includegraphics[width=0.7\textwidth]{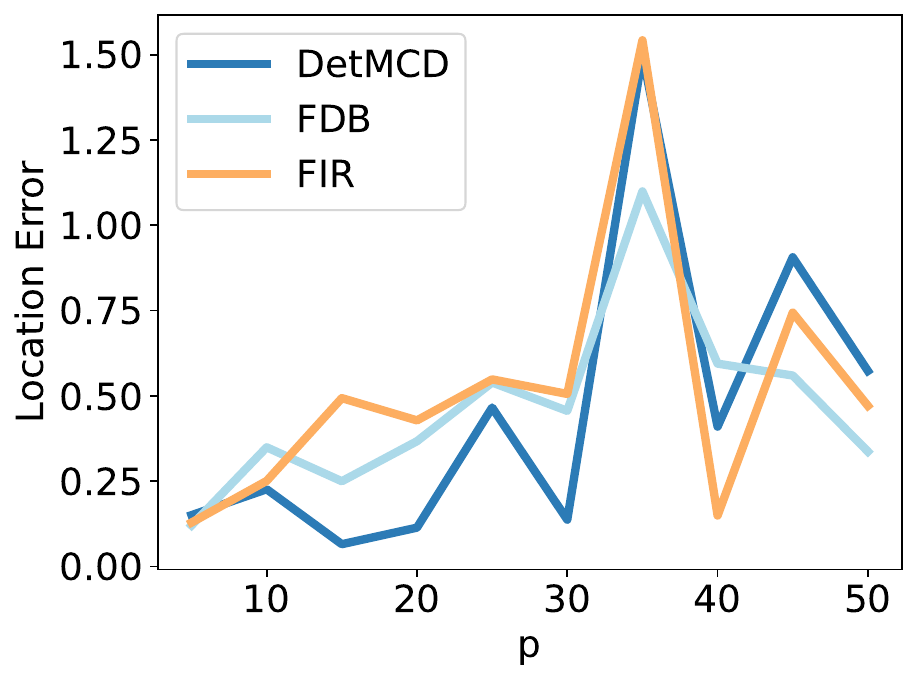}
    \label{fig:LocationError_p_A_cluster_n1000}
    \end{subfigure}
    \begin{subfigure}[t]{0.32 \textwidth}
    \includegraphics[width=0.7\textwidth]{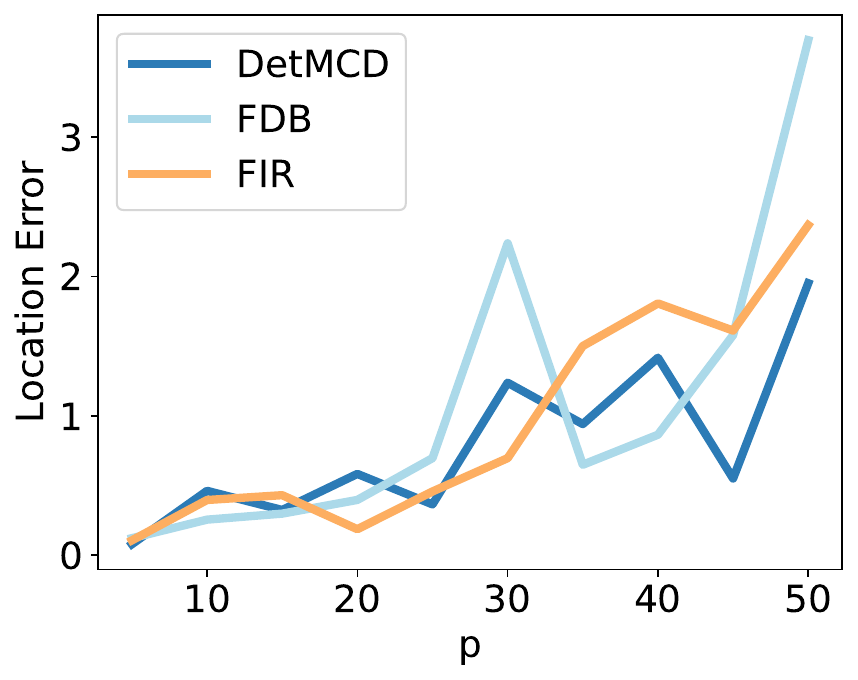}
    \label{fig:LocationError_p_A_radial_n1000}
    \end{subfigure}
    \begin{subfigure}[t]{0.32 \textwidth}
    \includegraphics[width=0.7\textwidth]{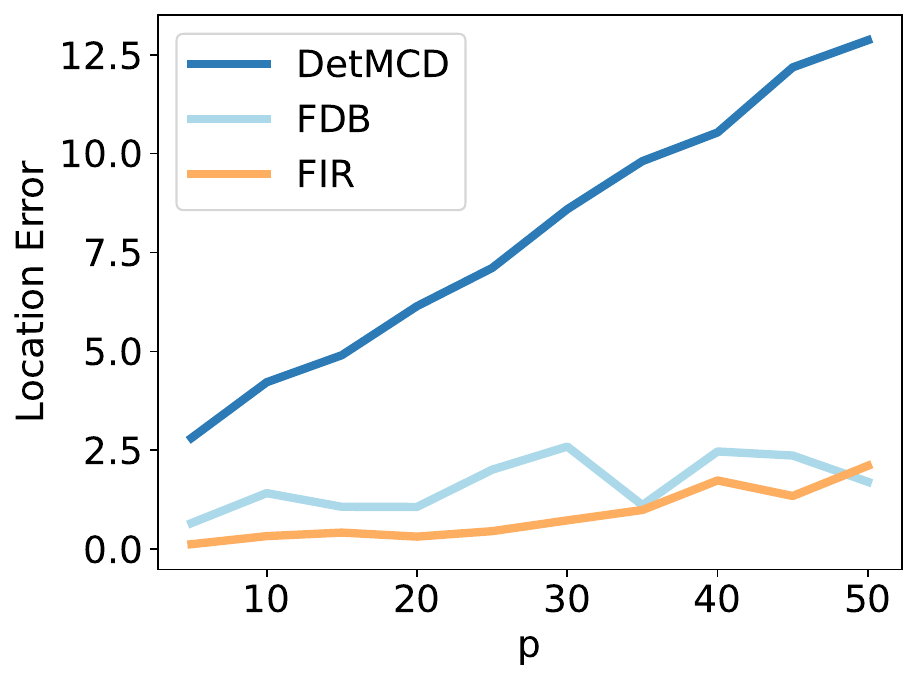}
    \label{fig:LocationError_p_A_point_n1000}
    \end{subfigure}
    \begin{subfigure}[t]{0.32 \textwidth}
    \includegraphics[width=0.7\textwidth]{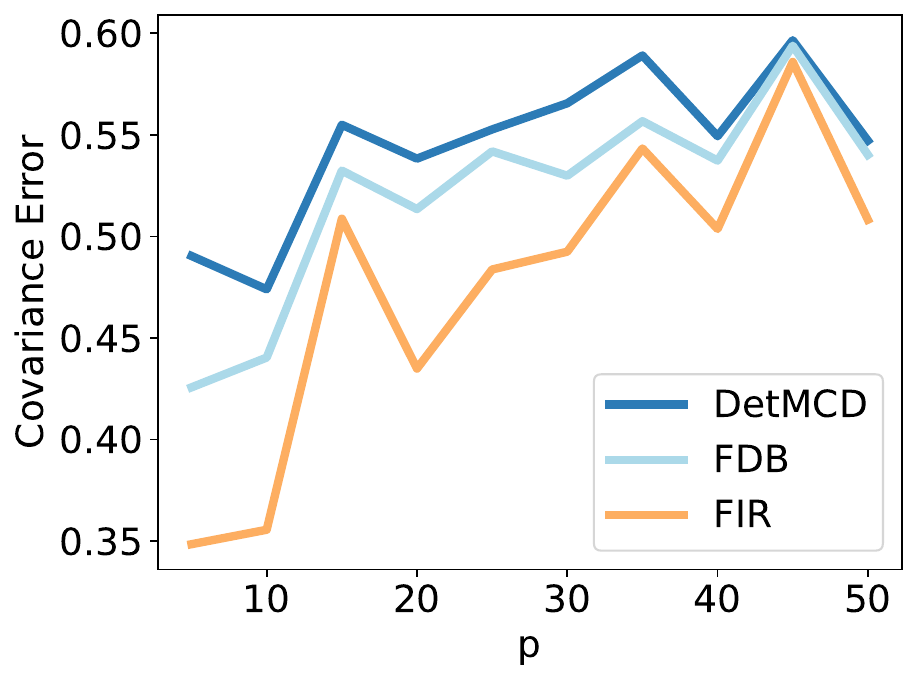}
    \label{fig:CovError_p_A_cluster_n1000}
    \end{subfigure}
    \begin{subfigure}[t]{0.32 \textwidth}
    \includegraphics[width=0.7\textwidth]{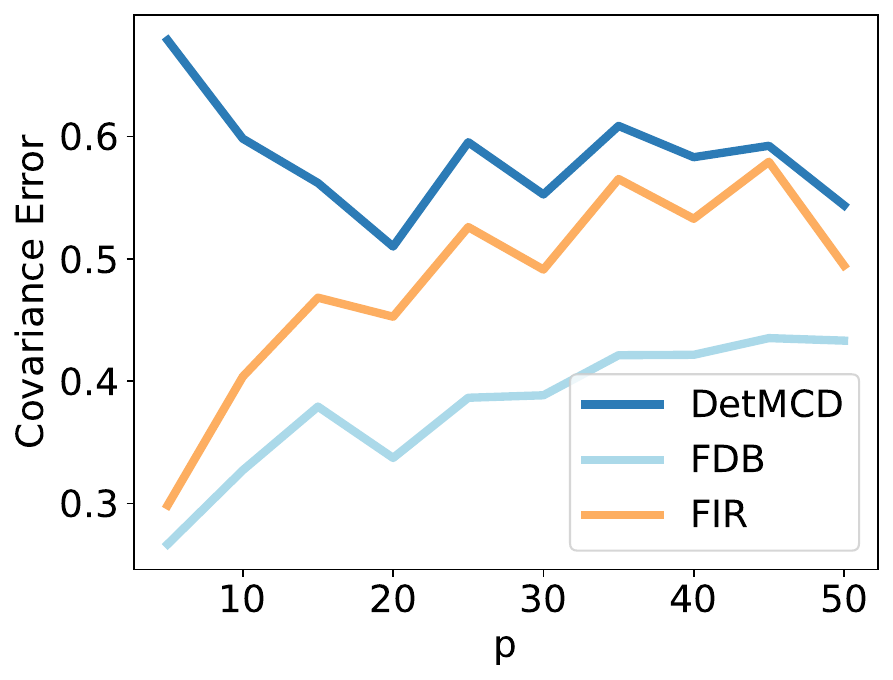}
    \label{fig:CovError_p_A_radial_n1000}
    \end{subfigure}
    \begin{subfigure}[t]{0.32 \textwidth}
    \includegraphics[width=0.7\textwidth]{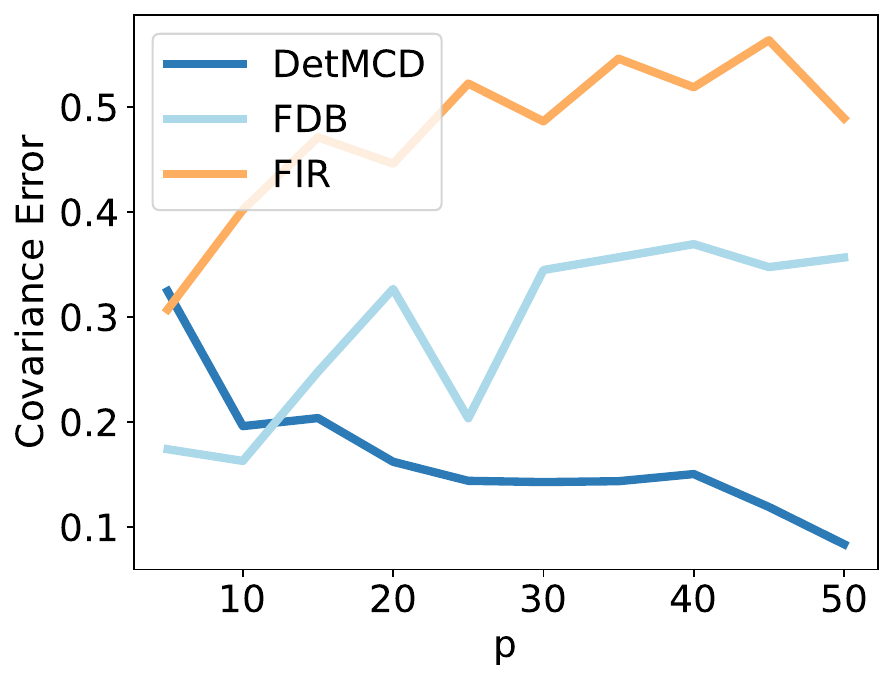}
    \label{fig:CovError_p_A_point_n1000}
    \end{subfigure}
    \begin{subfigure}[t]{0.32 \textwidth}
    \includegraphics[width=0.7\textwidth]{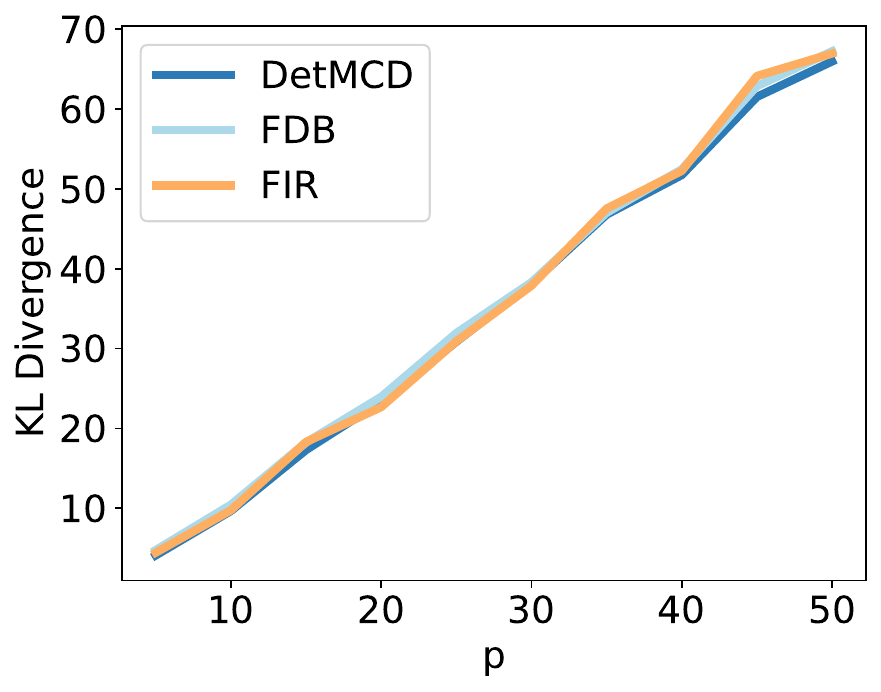}
    \caption{Cluster}
    \label{fig:KLDivergence_p_A_cluster_n1000}
    \end{subfigure}
    \begin{subfigure}[t]{0.32 \textwidth}
    \includegraphics[width=0.7\textwidth]{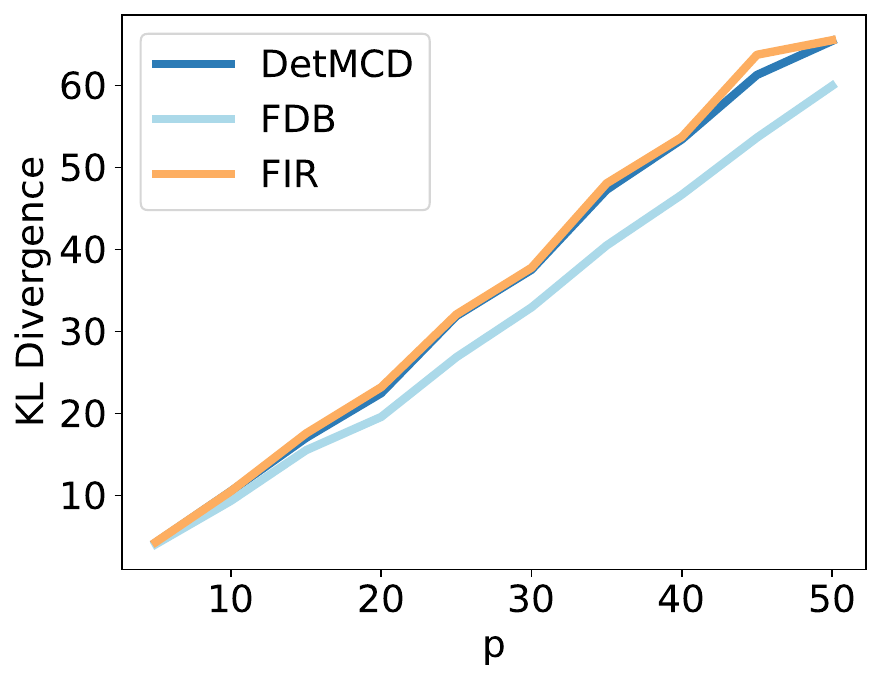}
    \caption{Radial}
    \label{fig:KLDivergence_p_A_radial_n1000}
    \end{subfigure}
    \begin{subfigure}[t]{0.32 \textwidth}
    \includegraphics[width=0.7\textwidth]{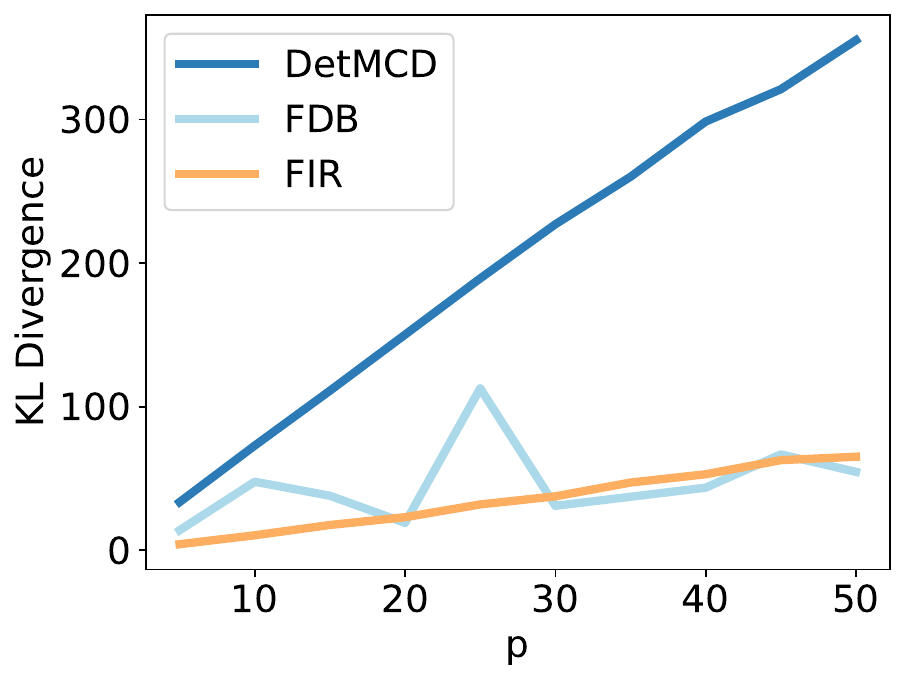}
    \caption{Point}
    \label{fig:KLDivergence_p_A_point_n1000}
    \end{subfigure}
    \caption{Approximation errors with varying number of features $p$. The examples uses $n=1000$ samples with $10\%$ outlier contamination and $\alpha = 0.75$. The first, second, and third row corresponds to the location error, covariance error, and KL-divergence, respectively}
    \label{fig:approx-errors-various-p}
\end{figure}
\begin{figure}[H]
    \centering
    \begin{subfigure}[t]{0.32 \textwidth}
    \includegraphics[width=0.7\textwidth]{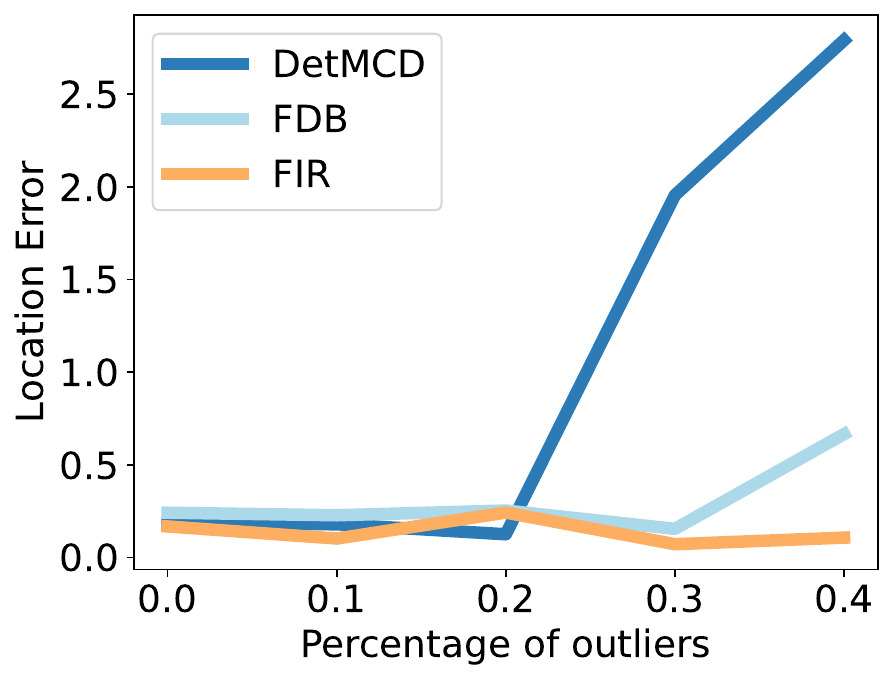}
    \label{fig:LocationError_Noise_level_A_point_p5_n1000}
    \end{subfigure}
    \begin{subfigure}[t]{0.32 \textwidth}
    \includegraphics[width=0.7\textwidth]{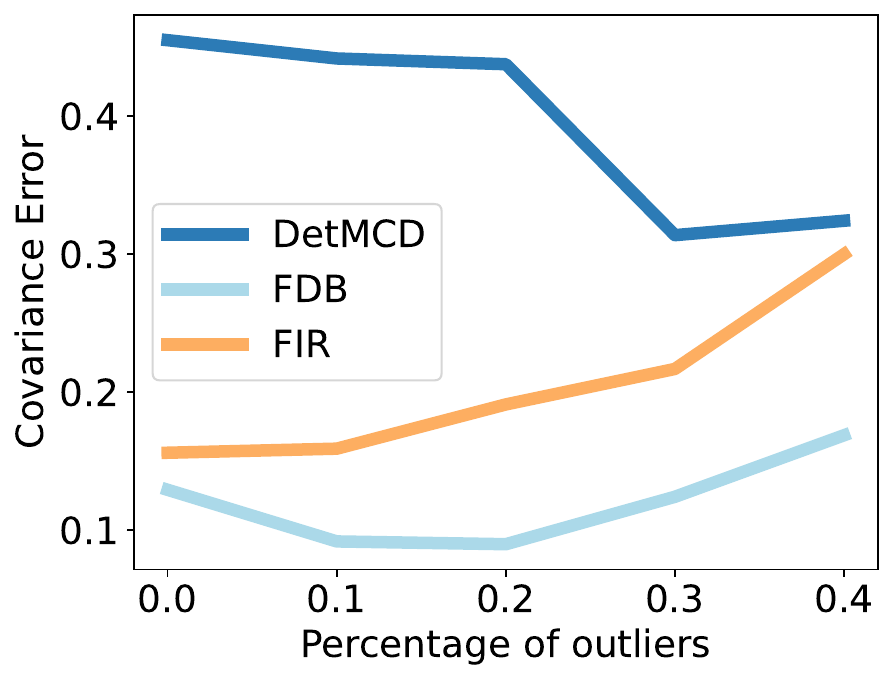}
    \label{fig:CovError_Noise_level_A_point_p5_n1000}
    \end{subfigure}
    \begin{subfigure}[t]{0.32 \textwidth}
    \includegraphics[width=0.7\textwidth]{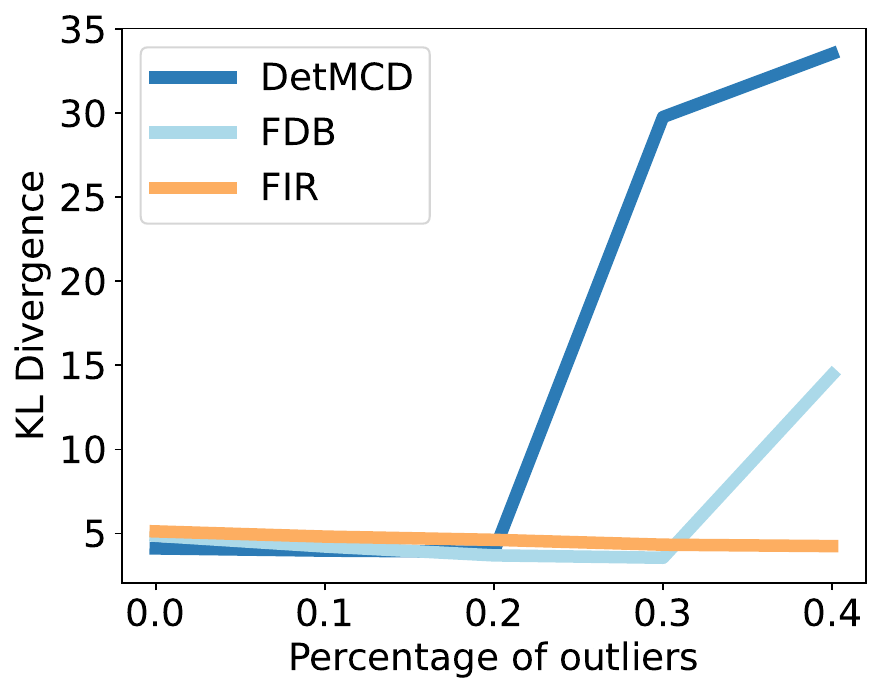}
    \label{fig:KLDivergence_Noise_level_A_point_p5_n1000}
    \end{subfigure}
    \caption{Approximation errors with varying outlier percentages. The example uses the \textbf{Point} outliers with $n=1000$ samples, $p=5$ features, and $\alpha = 0.5$. The location error, covariance error, and KL-divergence are shown from left to right.}
    \label{fig:approx-error-noise-level}
\end{figure}
\begin{figure}[H]
    \centering
    \begin{subfigure}[t]{0.24 \textwidth}
        \includegraphics[width=\textwidth]{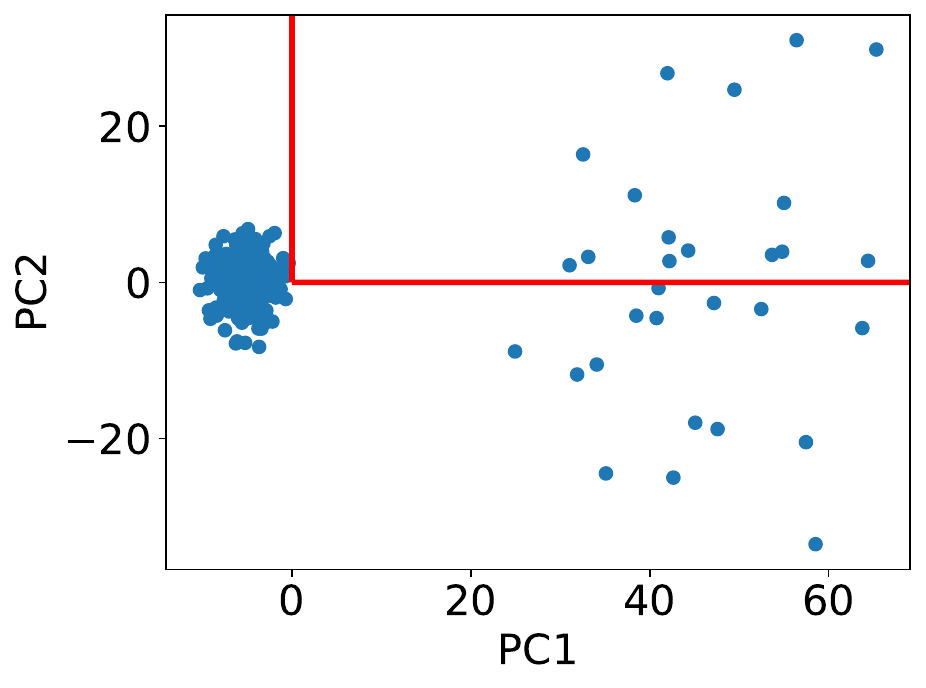}
        \caption{CPCA}
        \label{subfig:PCA}
    \end{subfigure}
    \begin{subfigure}[t]{0.24 \textwidth}
        \includegraphics[width=\textwidth]{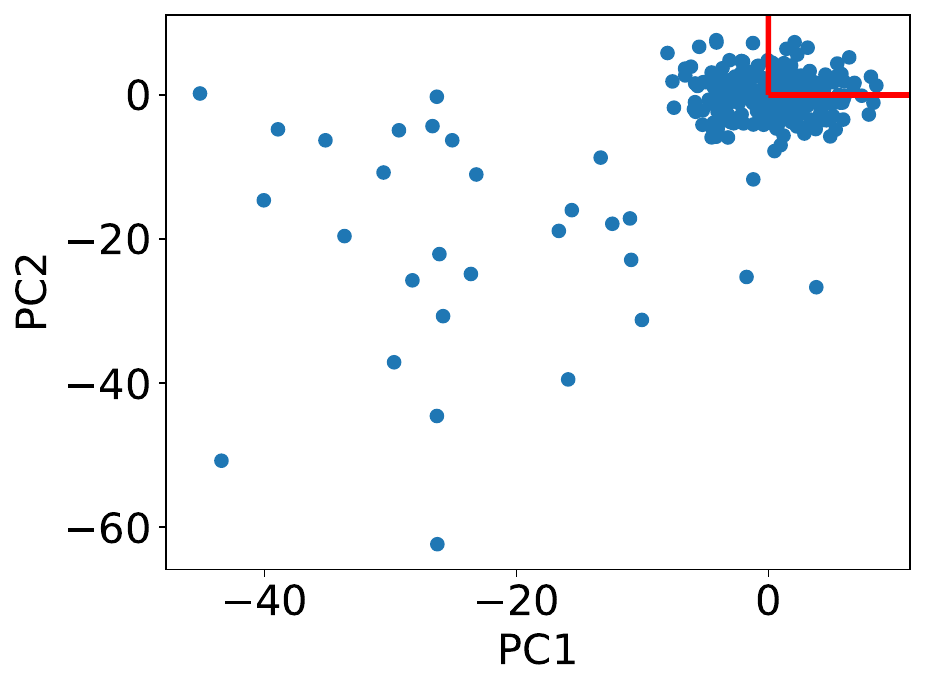}
        \caption{DetMCD-PCA}
        \label{subfig:robPCA}
    \end{subfigure}
    \begin{subfigure}[t]{0.24 \textwidth}
        \includegraphics[width=\textwidth]{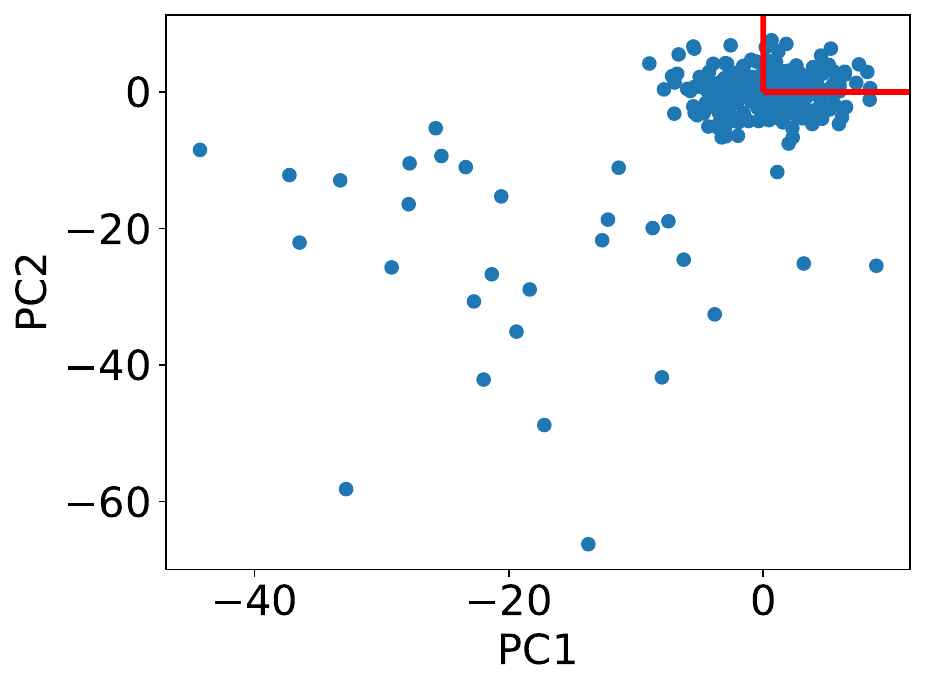}
        \caption{FDB-PCA}
        \label{subfig:FDB_PCA}
    \end{subfigure}
    \begin{subfigure}[t]{0.24 \textwidth}
        \includegraphics[width=\textwidth]{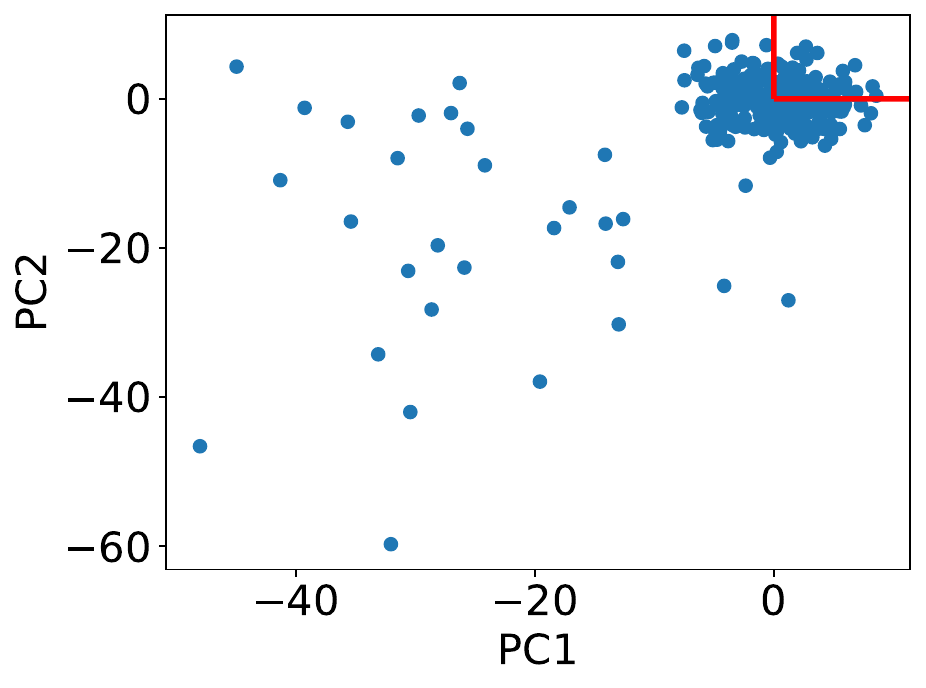}
        \caption{FIR-PCA}
        \label{subfig:FIR_PCA}
    \end{subfigure}
    \caption{CPCA, DetMCD-PCA, FDB-PCA, and FIR-PCA projections of $\mathbf{U V}$ with injected outliers. The red lines depict the magnitudes of variance in PC1 and PC2.}
    \label{fig:robust-PCA-comparison}
\end{figure}

\subsection{Performance Comparison}
\label{subsec:performance}
We utilize datasets of varying sizes to evaluate the computational runtime of the DetMCD, FDB, and FIR. The runtimes are measured in Python on a Dell Precision 5820 Tower workstation equipped with Intel Xeon W-2255 with core processing units of frequency 3.70 GHz, and a Nvidia RTX A4000 GPU. The results reported correspond to the average $100$ runs.

 \autoref{subfig:Time_n_A_2000_point_p5}, \autoref{subfig:Time_n_A_2000_point_p40}, and \autoref{subfig:Time_p_A_point_n1000} shows the runtimes of DetMCD, FDB, and FIR for $p=5$, $p=40$ and varying $p$ values, respectively. The DetMCD is significantly more expensive than both the FDB and FIR methods. The DetMCD requires concentration steps that are computationally expensive. The Bottom row provides a version of the top row without DetMCD to allow a better comparison of the FDB and FIR methods. Our proposed method requires slightly more time because it's based on calculating the projection depth, which is also used in FDB and the iterative steps described in \autoref{algo:fir-algo}. From our experiments, the FIR is at most $2\times$ slower than the FDB method. 

\autoref{tab:batch-size-effect} shows the execution time and accuracy for the FIR method with varying batch sizes. As the batch size increases, the performance increases. The batch size must be larger than $p$ and smaller than the number of inliers ($p < batch < \alpha n$). As the batch size gets closer to $h=\alpha n$ the accuracy deteriorates. The choice of batch size depends on the type of data and the computational resources. For the experiments in this paper, a batch size of approximately $10\%$ of the data leads to significant improvement in performance without degrading the accuracy. The appropriate batch size should maximize the performance while ensuring the desired accuracy. In cases where performance is no major concern, a conservative choice of batch size that is closer to $p$ yields more robust and accurate approximations. 
\begin{figure}[H]
    \centering
    \begin{subfigure}[t]{0.32 \textwidth}
        \includegraphics[width=0.7\textwidth]{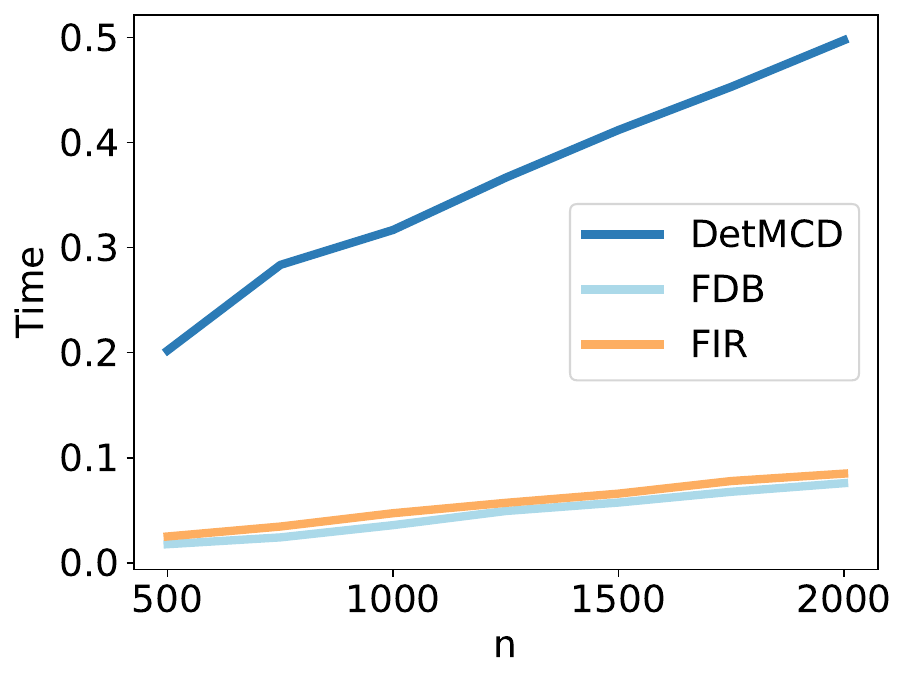}
        \label{fig:Time_n_all_A_2000_point_p5}
    \end{subfigure}
    \begin{subfigure}[t]{0.32 \textwidth}
        \includegraphics[width=0.7\textwidth]{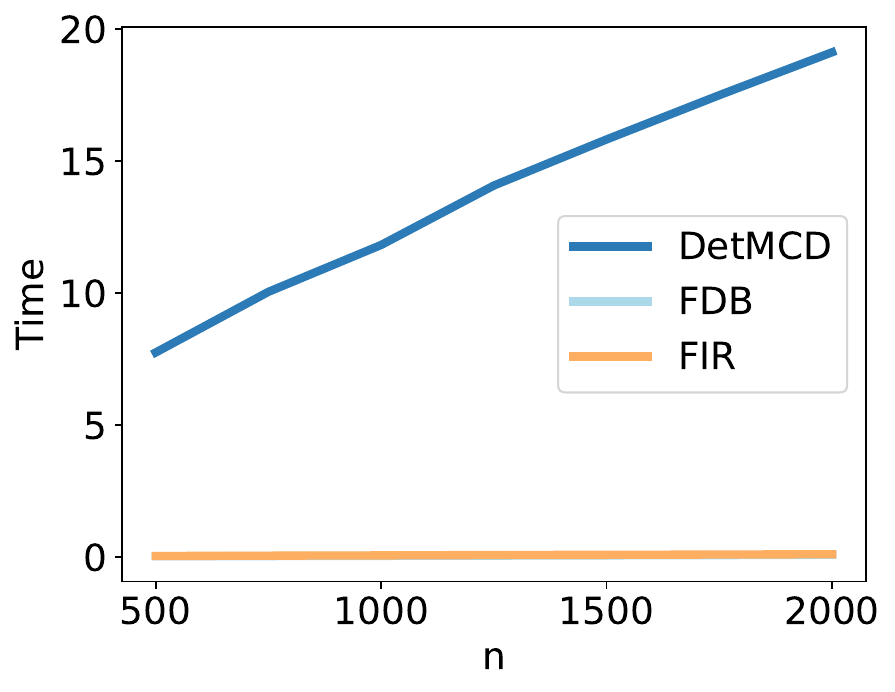}
        \label{fig:Time_n_all_A_2000_point_p5}
    \end{subfigure}
    \begin{subfigure}[t]{0.32 \textwidth}
        \includegraphics[width=0.7\textwidth]{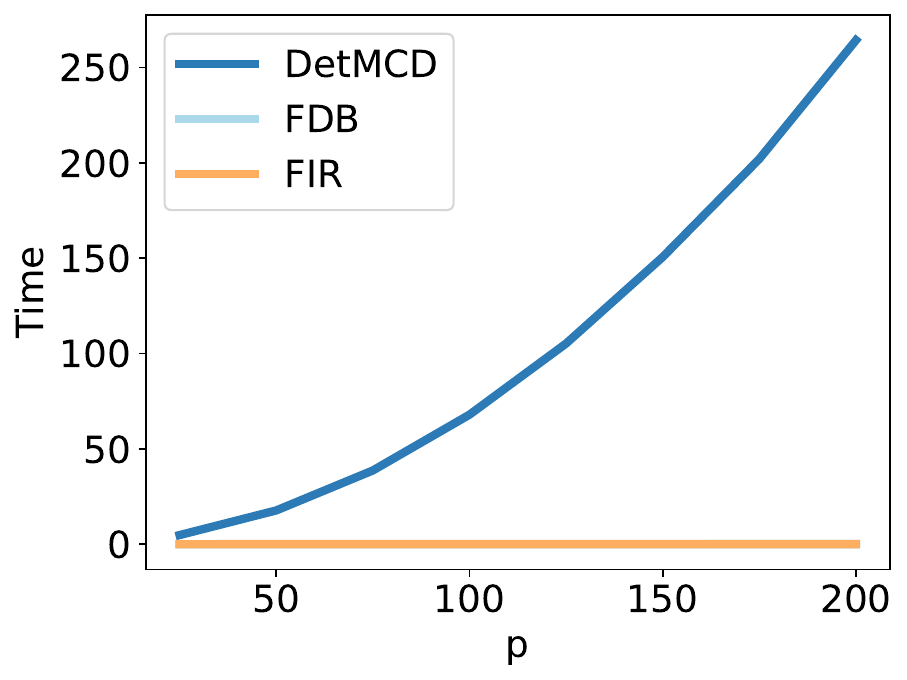}
        \label{subfig:Time_p_all_A_point_n1000}
    \end{subfigure}
    \begin{subfigure}[t]{0.32 \textwidth}
        \includegraphics[width=0.7\textwidth]{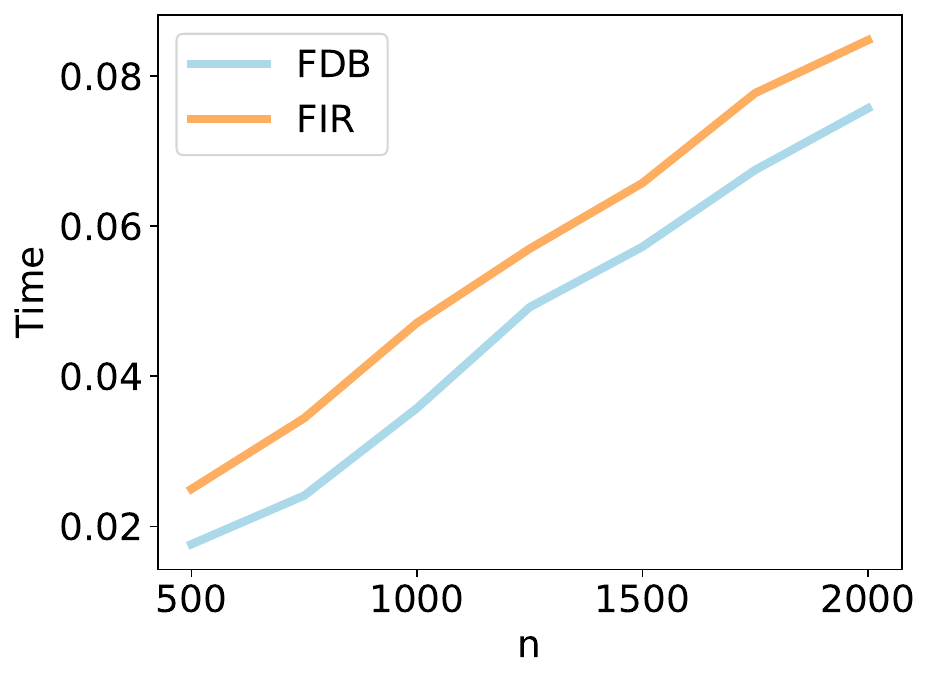}
        \caption{$p=5$}
        \label{subfig:Time_n_A_2000_point_p5}
    \end{subfigure}
    \begin{subfigure}[t]{0.32 \textwidth}
        \includegraphics[width=0.7\textwidth]{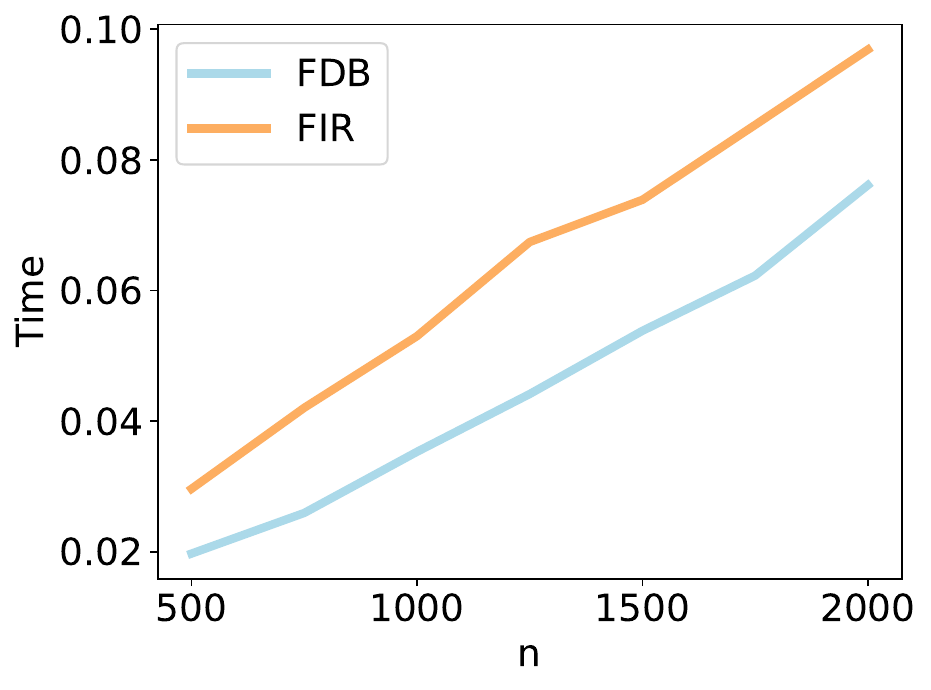}
        \caption{$p=40$}
        \label{subfig:Time_n_A_2000_point_p40}
    \end{subfigure}
    \begin{subfigure}[t]{0.32 \textwidth}
        \includegraphics[width=0.7\textwidth]{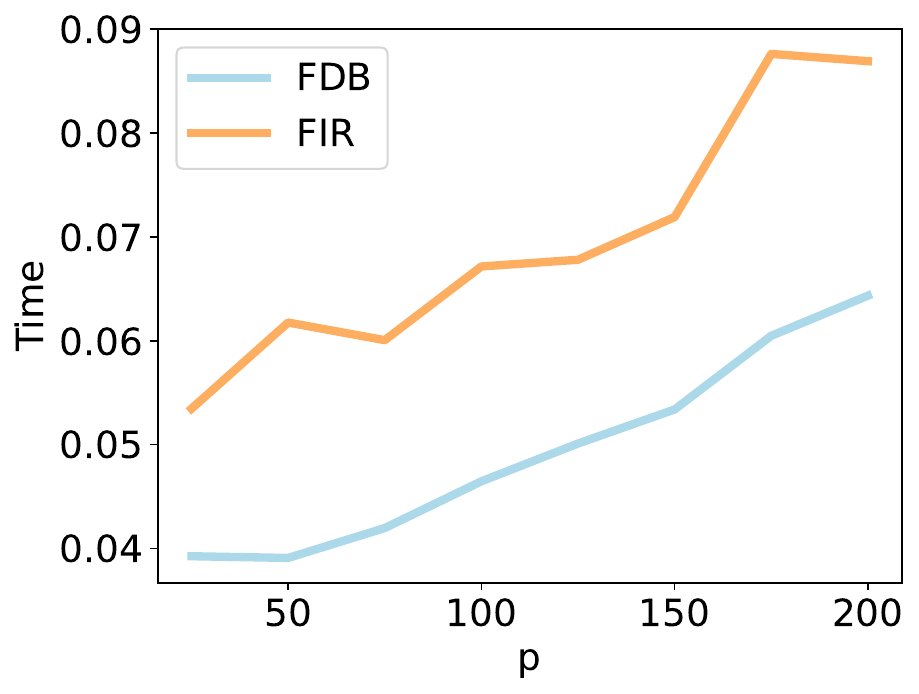}
        \caption{$n=1000$}
        \label{subfig:Time_p_A_point_n1000}
    \end{subfigure}
    \caption{Execution time in $s$ of DetMCD, FDB, and FIR for varying $n$ with $p=5, 40$, and varying $p$ with $n=1000$. The bottom row shows a visualization of the top row without DetMCD.}
    \label{fig:runtimes}
\end{figure}
\begin{table}[H]
    \caption{Effect of batch size on FIR performance and accuracy averaged over $1000$ sampled datasets with $10\%$ \textbf{Point} outliers.}
    \label{tab:batch-size-effect}
    \centering
    \scriptsize
    \begin{tabular}{lcccccccccl}
    \toprule
& \multicolumn{2}{c}{ 200 }& \multicolumn{2}{c}{ 400 }& \multicolumn{2}{c}{ 800 }& \multicolumn{2}{c}{ 1600 }& \multicolumn{2}{c}{ 3200 }\\
    \cmidrule(r){2-3}\cmidrule(r){4-5}\cmidrule(r){6-7}\cmidrule(r){8-9}\cmidrule{10-11}
batch size  & $t$ & $e_{\mu}$  & $t$ & $e_{\mu}$  & $t$ & $e_{\mu}$  & $t$ & $e_{\mu}$  & $t$ & $e_{\mu}$ \\
10 & 0.021 & 0.135 & 0.041 & 0.297 & 0.081 & 0.450 & 0.167 & 0.407 & 0.366 & 0.094\\
40 & 0.013 & 0.221 & 0.025 & 0.349 & 0.045 & 0.423 & 0.087 & 0.413 & 0.177 & 0.097\\
60 & 0.012 & 0.192 & 0.022 & 0.242 & 0.041 & 0.389 & 0.078 & 0.412 & 0.156 & 0.098\\
120 & 0.011 & 0.232 & 0.021 & 0.321 & 0.037 & 0.506 & 0.069 & 0.357 & 0.134 & 0.087\\
240 & -- & -- & 0.020 & 0.505 & 0.035 & 0.437 & 0.065 & 0.415 & 0.124 & 0.106\\
300 & -- & -- & 0.019 & 0.058 & 0.035 & 0.217 & 0.064 & 0.337 & 0.122 & 0.098\\
\bottomrule
\end{tabular}
\end{table}

\section{Real Data Examples}
\label{sec:results-real-world}

\subsection{Octane Dataset}
\label{subsec:octane}
The octane dataset, used in ~\cite{hubert_robpca_2005} has near-infrared (NIR) absorbance spectra over $p = 226$ wavelengths of $n =39$ gasoline samples with specific octane numbers. The six samples $25$, $26$, and $36–39$ are known to contain added alcohol. Although $p=226$, the first two principal components are sufficient to describe more than $80\%$ of the data. Following the approach in ~\cite{hubert_robpca_2005}, the horizontal and vertical cut-off values, shown by the red dashed lines, are calculated according to $\sqrt{\chi^{2}_{q, .975}}$, $\mu + \sigma z_{.975}$ with $z_{.975}$ as the $97.5\%$ quantile of the Gaussian distribution, $\chi^{2}$ denotes the chi-squared distribution, $q$ the retained number of components. The execution of CPCA, DetMCD-PCA, FDB-PCA, and FIR-PCA is $1.55$ ms, $56.57$ ms, $3.90$ ms, and $5.03$ ms, respectively. The red dots in \autoref{fig:octane} indicate outliers that are well separated from the inliers and greater than the cut-off values. The six samples are not easily identified with the CPCA outlier map, as shown in the first column of \autoref{fig:octane}. However, the DetMCD, FDB-PCA, and FIR-PCA separate these outliers as demonstrated in the outlier maps and scores visualization in \autoref{fig:octane}.
\begin{figure}[H]
    \centering
    \begin{subfigure}[t]{0.24 \textwidth}
        \includegraphics[width=\textwidth]{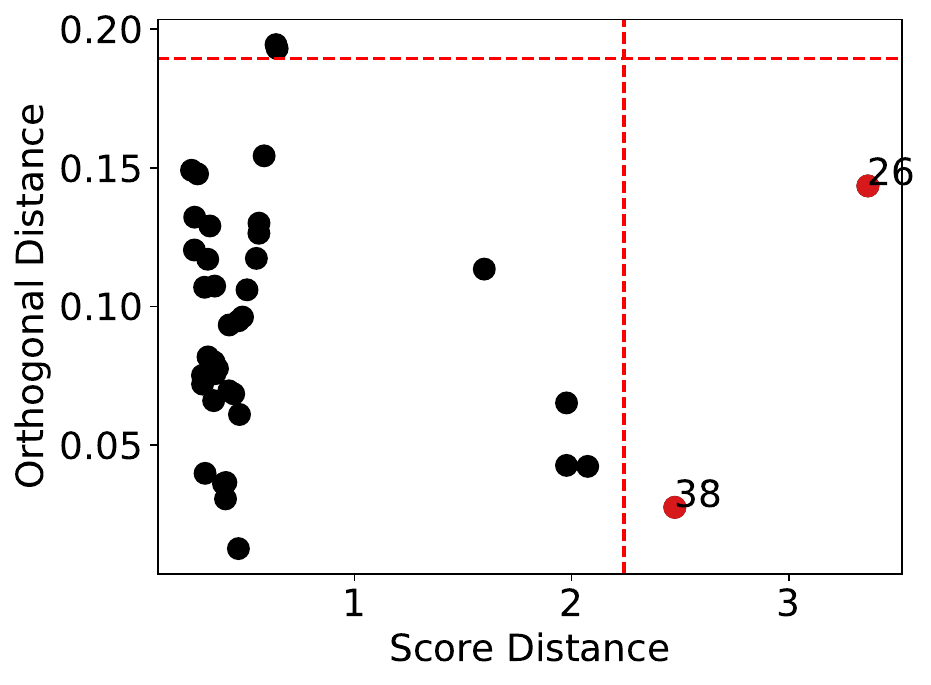}
        \caption{CPCA}
        \label{subfig:octane_C_PCA_outlier_map}
    \end{subfigure}
    \begin{subfigure}[t]{0.24 \textwidth}
        \includegraphics[width=\textwidth]{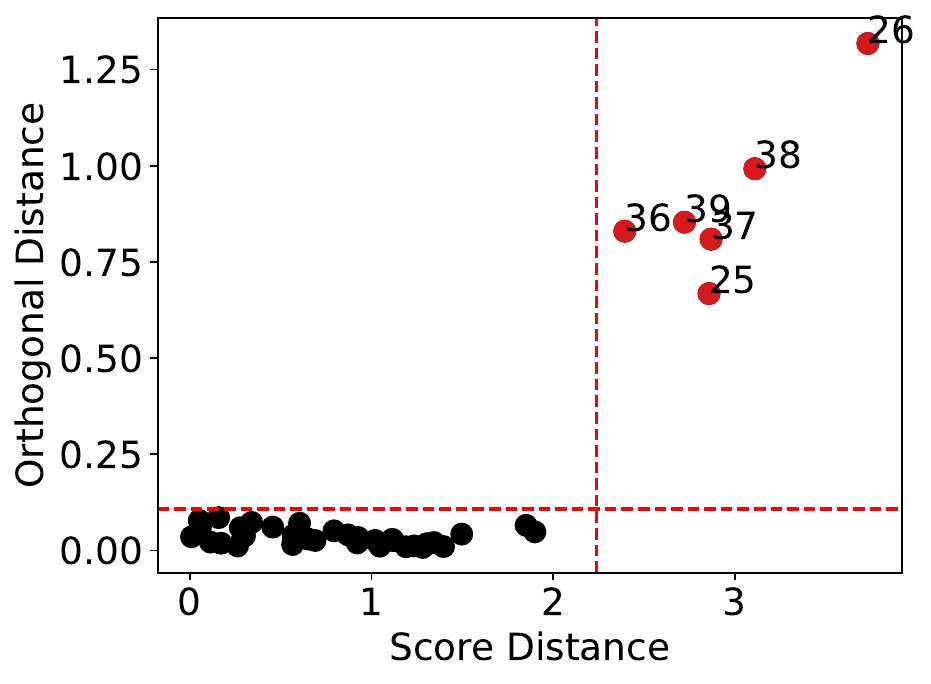}
        \caption{DetMCD-PCA}
        \label{subfig:octane_DetMCD_PCA_outlier_map}
    \end{subfigure}
    \begin{subfigure}[t]{0.24 \textwidth}
        \includegraphics[width=\textwidth]{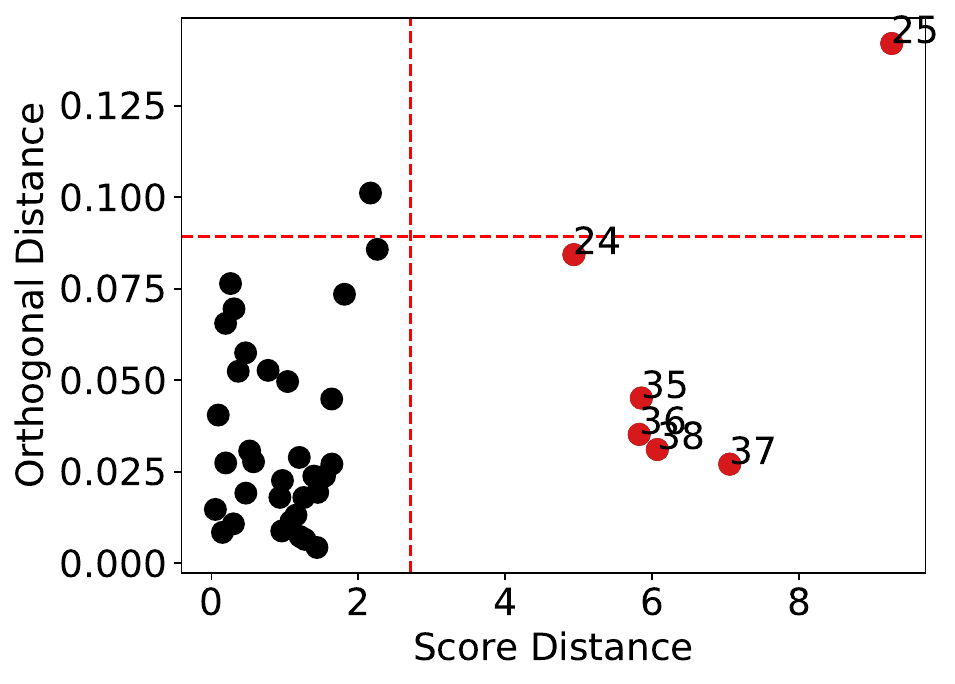}
        \caption{FDB-PCA}
        \label{subfig:octane_FDB_PCA_outlier_map}
    \end{subfigure}
    \begin{subfigure}[t]{0.24 \textwidth}
        \includegraphics[width=\textwidth]{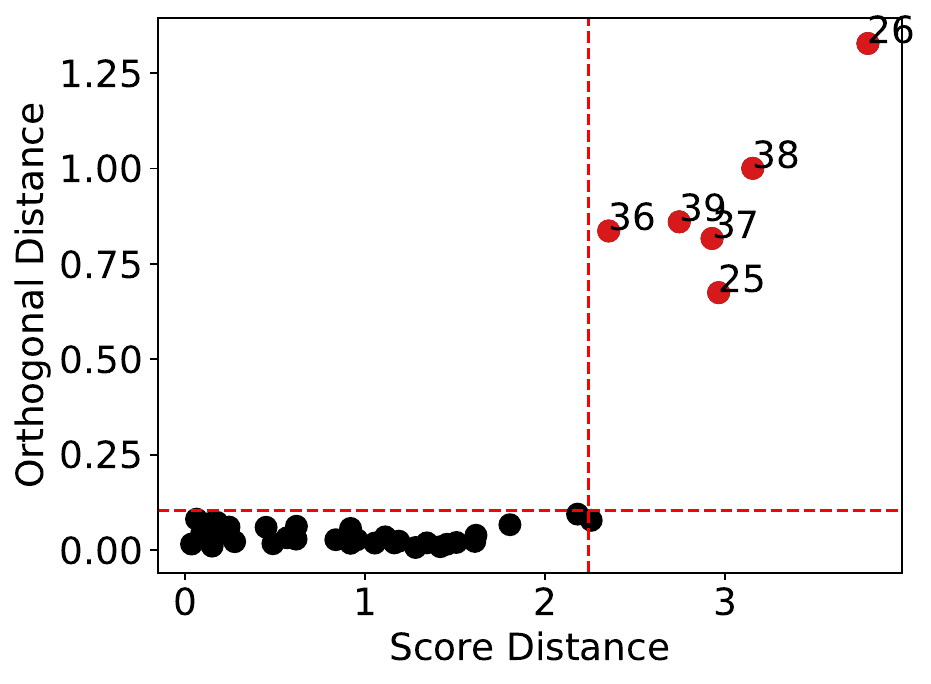}
        \caption{FIR-PCA}
        \label{subfig:octane_FIR_PCA_outlier_map}
    \end{subfigure}
    \caption{Score (top row) and outlier maps (bottom row) for CPCA, DetMCD-PCA, FIR-PCA, and FIR-PCA. The red points correspond to the outliers identified by each method. Dashed lines represent the outlier cutoff reported by the methods.}
    \label{fig:octane}
\end{figure}

\subsection{Forged Bank Notes}
This example, taken from ~\citep{Flury1988} has $p=6$ measurements and $n=100$ forged Swiss banknotes and prior studies have shown that this dataset has several outliers and correlated variables ~\citep{Willems01012009, Salibián-Barrera01092006,Hubert01072012}. The first two principal components explain more than $80\%$ of the data. The execution of CPCA, DetMCD-PCA, FDB-PCA, and FIR-PCA are $5.67$ ms, $115.57$ ms, $6.17$ ms, and $11.43.03$ ms, respectively. Here, we consider outliers $13$, $23$, $61$, $71$, $80$ and $87$ identified in previous studies \citep{Hubert01072012,Zhang02012024}. The diagnostic plot from CPCA considers $13$ and $87$ to be inliers, as shown in the bottom plot of \autoref{subfig:ForgedBankNotes_C_PCA_outlier_map}. The remaining approaches provide a comparable analysis of the data and successfully identify the set of outliers considered.
\begin{figure}[H]
    \centering
    \begin{subfigure}[t]{0.24 \textwidth}
        \includegraphics[width=\textwidth]{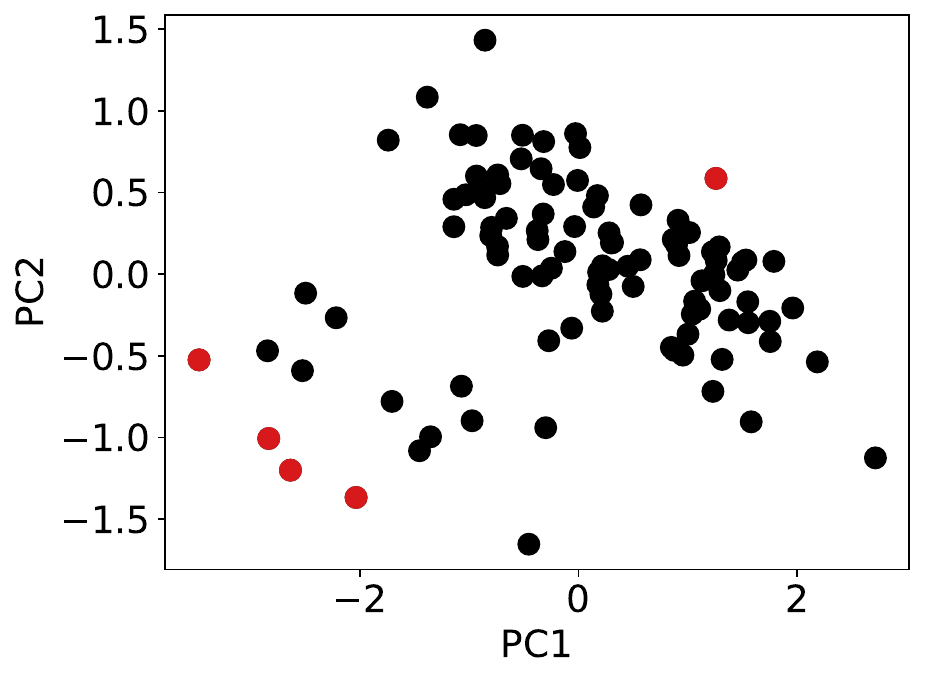}
        \label{subfig:ForgedBankNotes_C_PCA_first_two_principal_components}
    \end{subfigure}
    \begin{subfigure}[t]{0.24 \textwidth}
        \includegraphics[width=\textwidth]{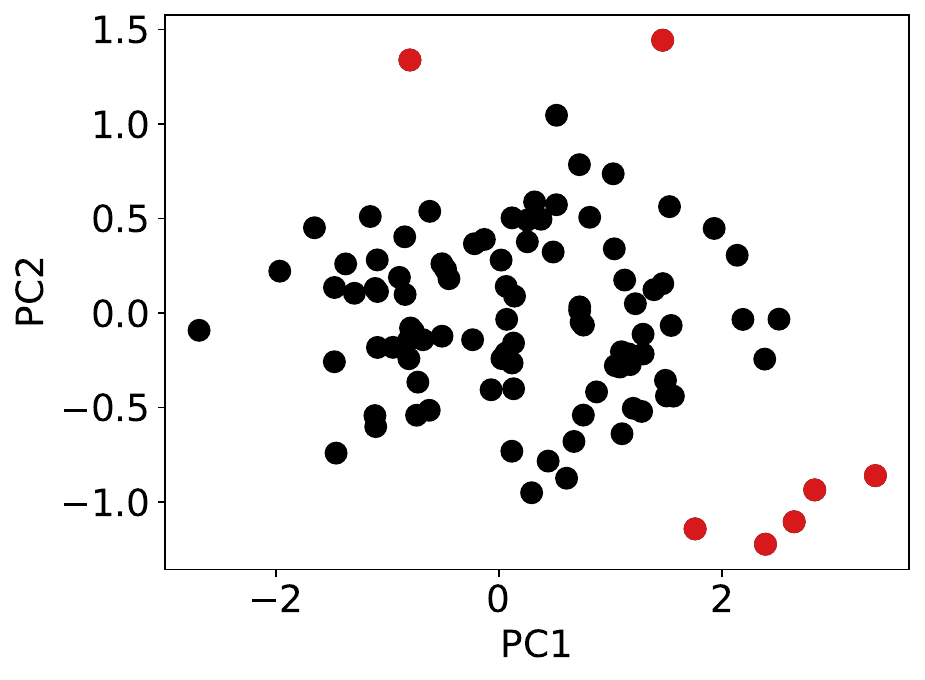}
        \label{subfig:ForgedBankNotes_DetMCD_PCA_first_two_principal_components}
    \end{subfigure}
    \begin{subfigure}[t]{0.24 \textwidth}
        \includegraphics[width=\textwidth]{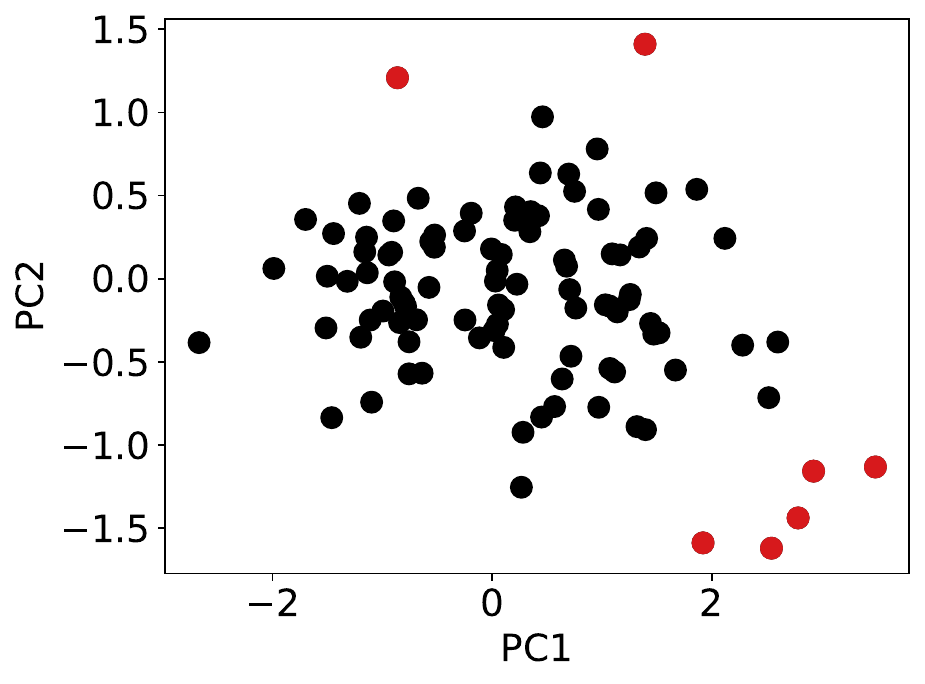}
        \label{subfig:ForgedBankNotes_FDB_PCA_first_two_principal_components}
    \end{subfigure}
    \begin{subfigure}[t]{0.24 \textwidth}
        \includegraphics[width=\textwidth]{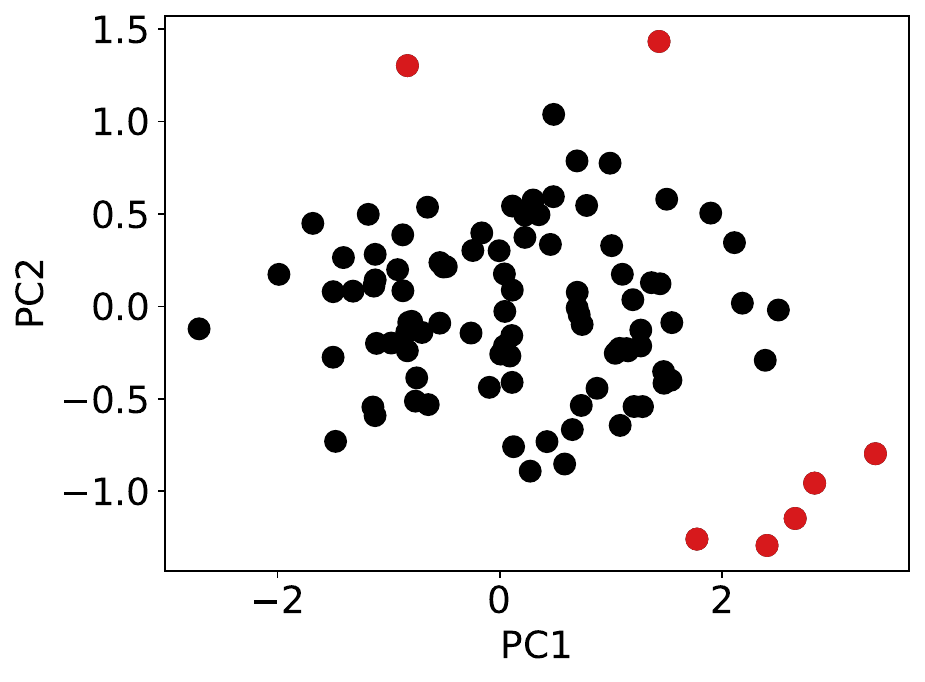}
        \label{subfig:ForgedBankNotes_FIR_PCA_first_two_principal_components}
    \end{subfigure}
    \begin{subfigure}[t]{0.24 \textwidth}
        \includegraphics[width=\textwidth]{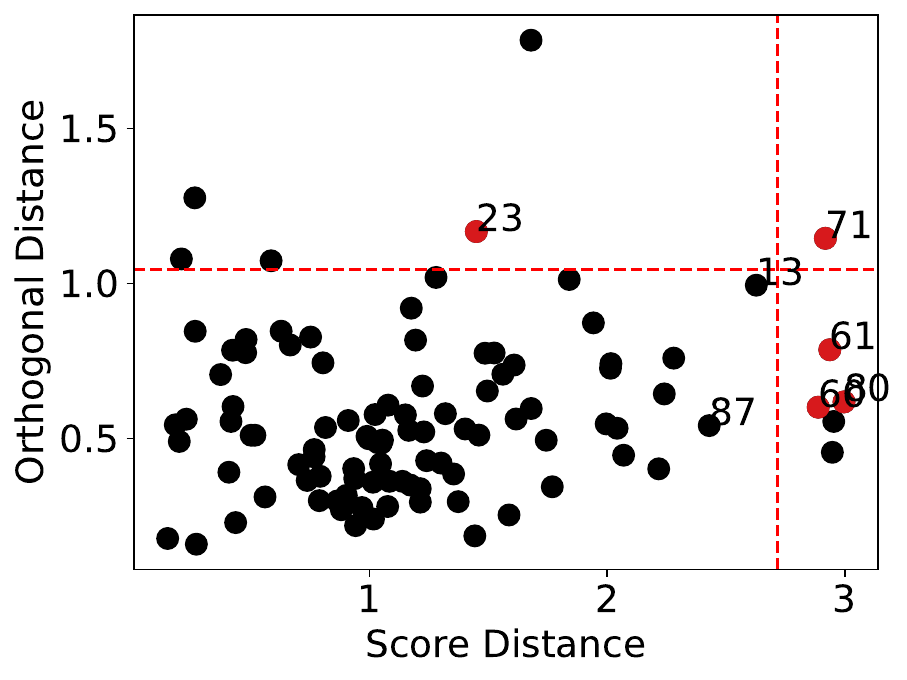}
        \caption{CPCA}
        \label{subfig:ForgedBankNotes_C_PCA_outlier_map}
    \end{subfigure}
    \begin{subfigure}[t]{0.24 \textwidth}
        \includegraphics[width=\textwidth]{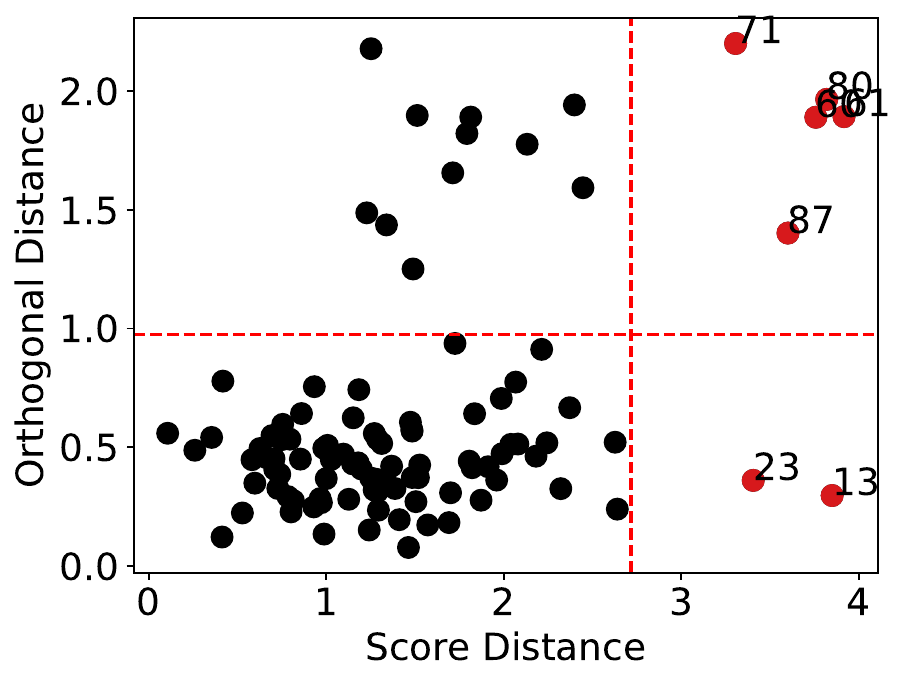}
        \caption{DetMCD-PCA}
        \label{subfig:ForgedBankNotes_DetMCD_PCA_outlier_map}
    \end{subfigure}
    \begin{subfigure}[t]{0.24 \textwidth}
        \includegraphics[width=\textwidth]{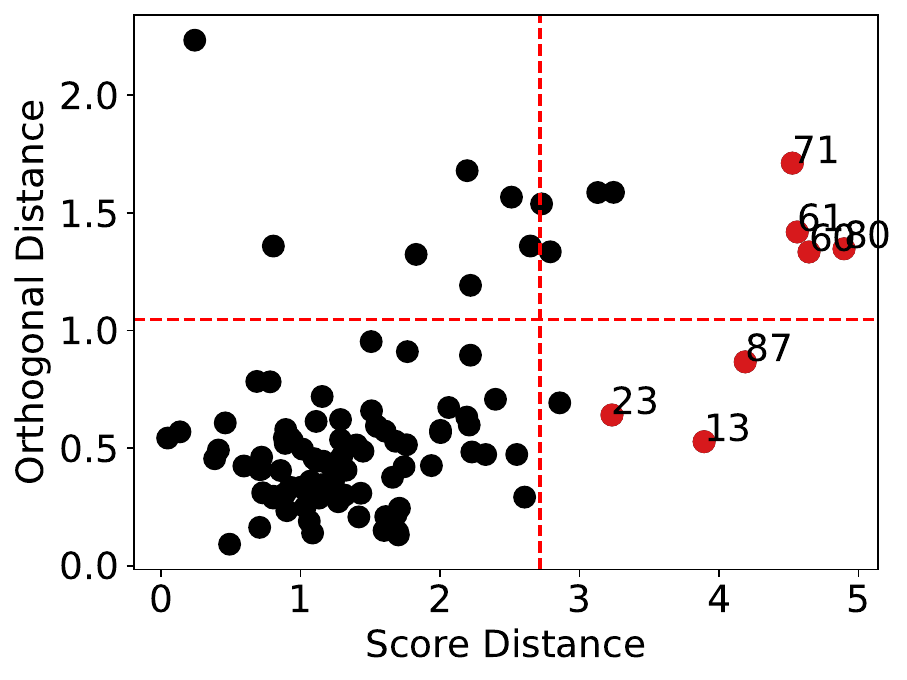}
        \caption{FDB-PCA}
        \label{subfig:ForgedBankNotes_FDB_PCA_outlier_map}
    \end{subfigure}
    \begin{subfigure}[t]{0.24 \textwidth}
        \includegraphics[width=\textwidth]{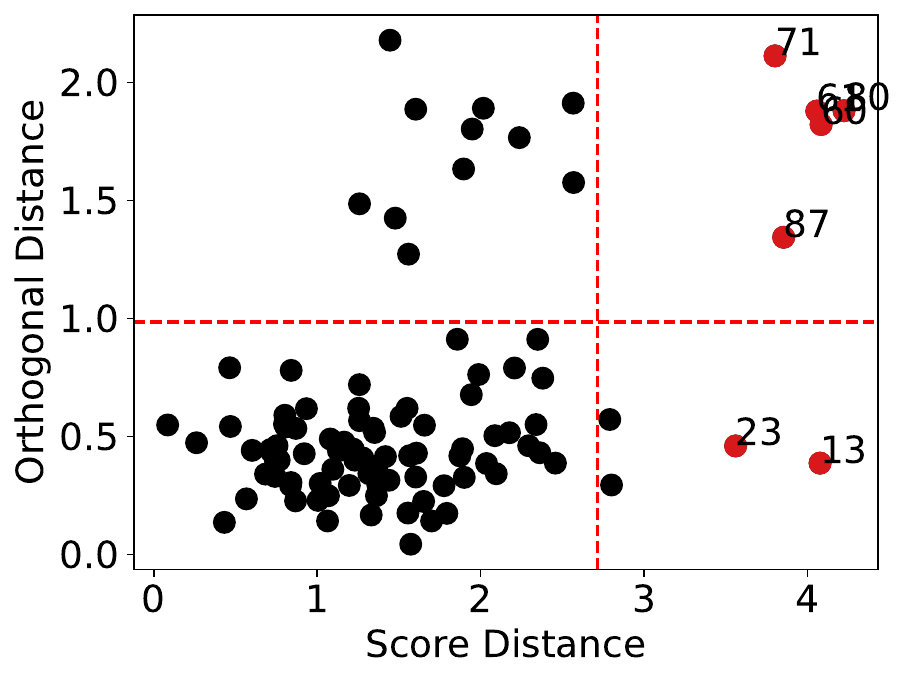}
        \caption{FIR-PCA}
        \label{subfig:ForgedBankNotes_FIR_PCA_outlier_map}
    \end{subfigure}
    \caption{Score (top row) and outlier maps (bottom row) for CPCA, DetcMCD-PCA, FIR-PCA, and FIR-PCA.The red points correspond to the outliers detected by each method. Dashed lines represent the outlier cutoff reported by the methods.}
    \label{fig:ForgedBankNotes}
\end{figure}

\subsection{TopGear}
\label{subsec:TipGear}
The TopGear dataset obtained from ~\cite{leyder2024robpy} contains different cars with performance information that were featured in the BBC television show Top Gear until 2014. We preprocess the original data as in ~\cite{leyder2024robpy} to remove non-numerical columns and rows, and columns with zero variance. The processed data has $p=10$ features and $n=245$ observations. The execution of CPCA, DetMCD-PCA, FDB-PCA, and FIR-PCA are $4.25$ ms, $403.92$ ms, $15.23$ ms, and $25.44$ ms, respectively. The top six outliers, marked by red dots, are detected by robust PCA methods, whereas the CPCA does not identify all six points as outliers, as shown by the outlier maps in the bottom row of \autoref{fig:topgear}. In \autoref{subfig:topgear_C_PCA_outlier_map}, $117$ is closer to the distribution, while in the remaining robust PCA outlier maps, $117$ is further away from the distribution. In this example, the points $53$ and $219$, shown in orange, have small score distances and large orthogonal distances. CPCA is sensitive to these outliers, as shown in \autoref{subfig:topgear_C_PCA_outlier_map}, where it incorrectly classifies the score distances as large and the orthogonal distances small.
\begin{figure}[H]
    \centering
    \begin{subfigure}[t]{0.24 \textwidth}
        \includegraphics[width=\textwidth]{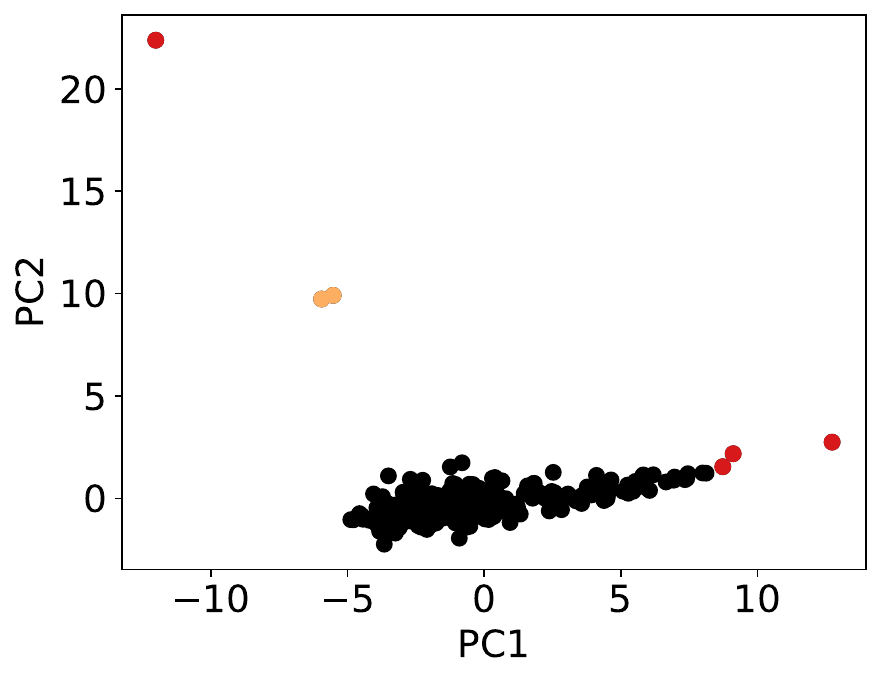}
        \label{subfig:topgear_C_PCA_first_two_principal_components}
    \end{subfigure}
    \begin{subfigure}[t]{0.24 \textwidth}
        \includegraphics[width=\textwidth]{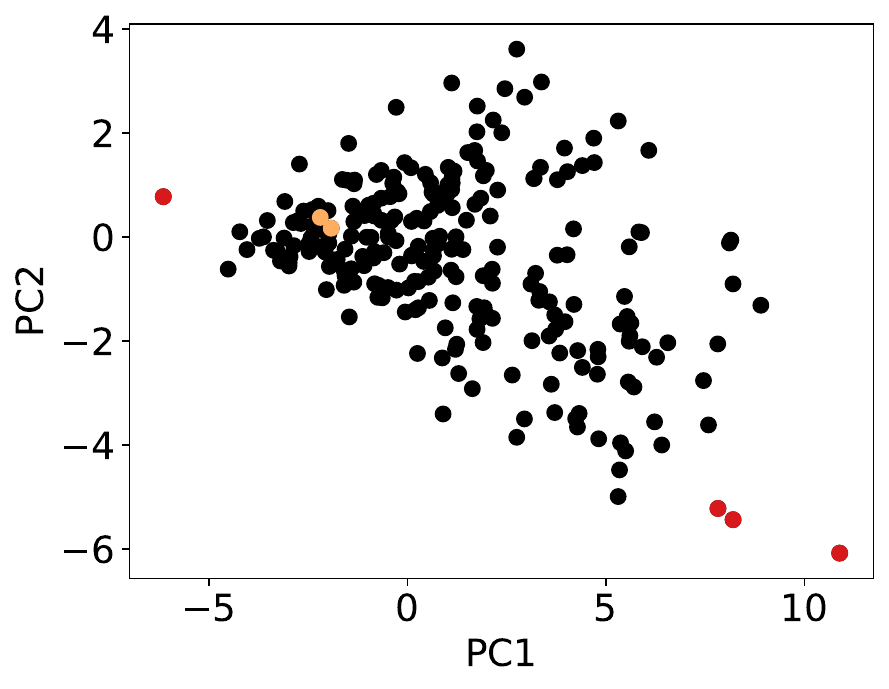}
        \label{subfig:topgear_DetMCD_PCA_first_two_principal_components}
    \end{subfigure}
    \begin{subfigure}[t]{0.24 \textwidth}
        \includegraphics[width=\textwidth]{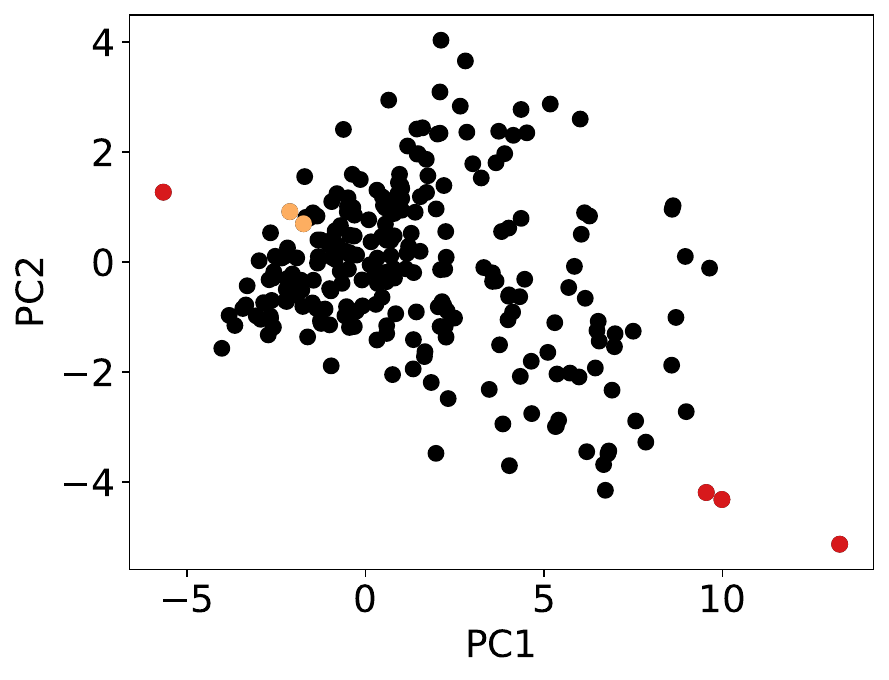}
        \label{subfig:topgear_FDB_PCA_first_two_principal_components}
    \end{subfigure}
    \begin{subfigure}[t]{0.24 \textwidth}
        \includegraphics[width=\textwidth]{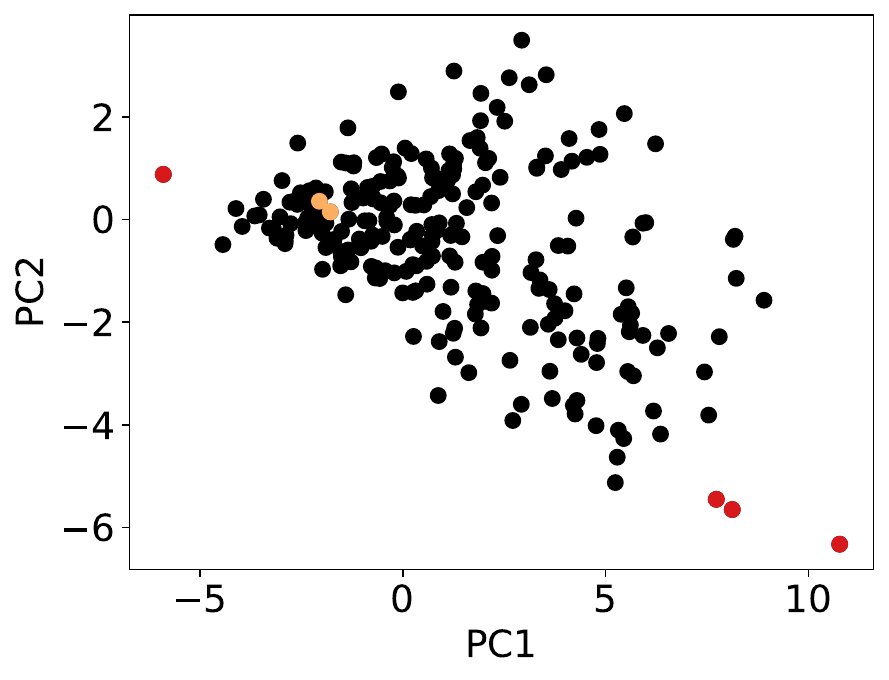}
        \label{subfig:topgear_FIR_PCA_first_two_principal_components}
    \end{subfigure}
    \begin{subfigure}[t]{0.24 \textwidth}
        \includegraphics[width=\textwidth]{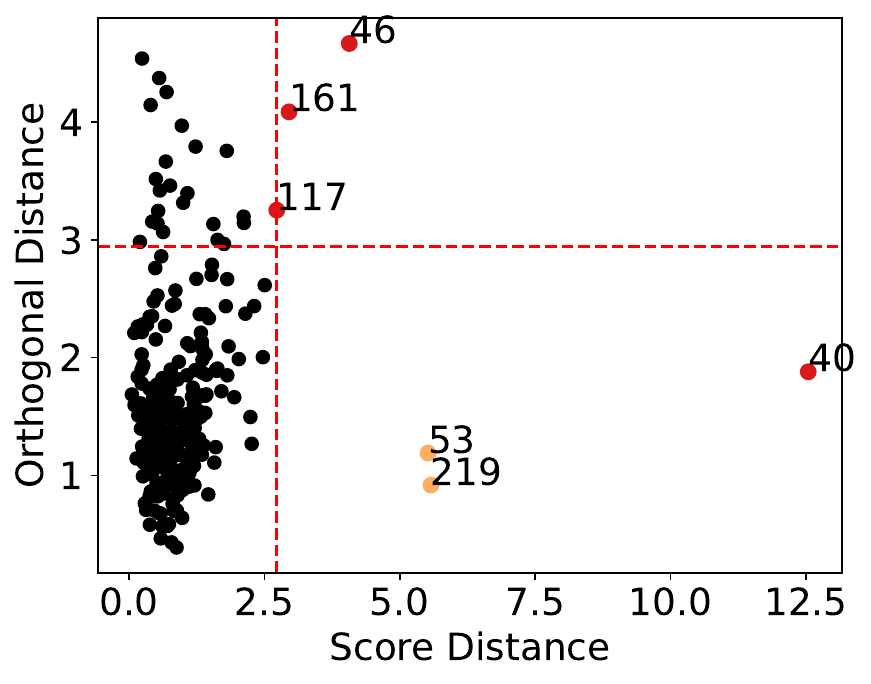}
        \caption{CPCA}
        \label{subfig:topgear_C_PCA_outlier_map}
    \end{subfigure}
    \begin{subfigure}[t]{0.24 \textwidth}
        \includegraphics[width=\textwidth]{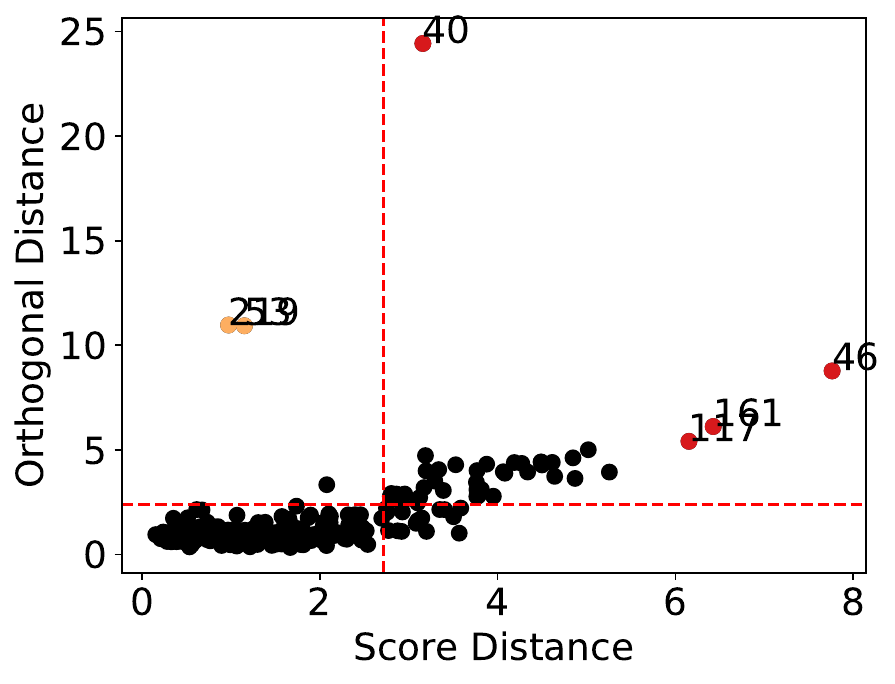}
        \caption{DetMCD-PCA}
        \label{subfig:topgear_DetMCD_PCA_outlier_map}
    \end{subfigure}
    \begin{subfigure}[t]{0.24 \textwidth}
        \includegraphics[width=\textwidth]{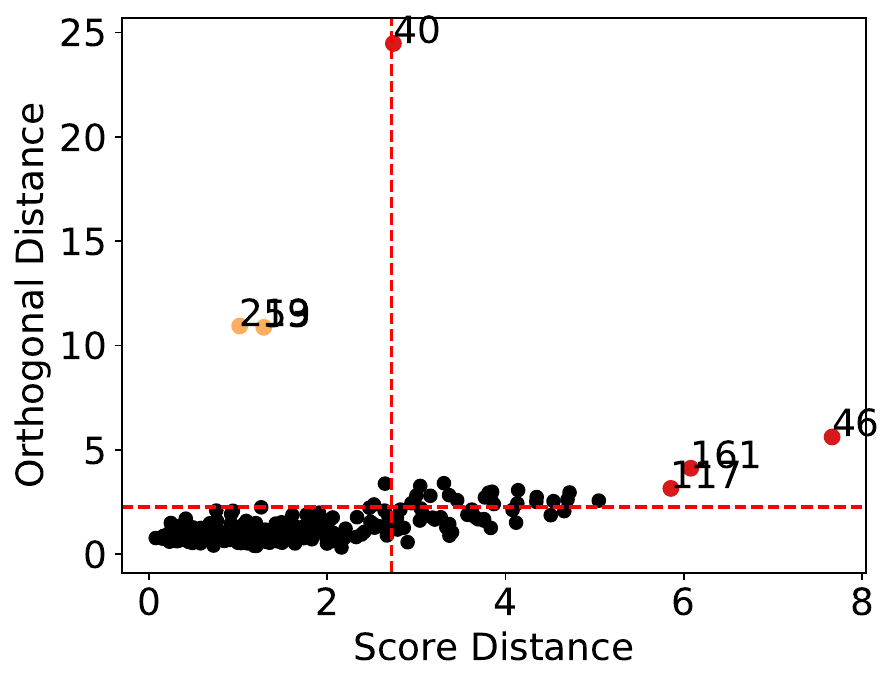}
        \caption{FDB-PCA}
        \label{subfig:topgear_FDB_PCA_outlier_map}
    \end{subfigure}
    \begin{subfigure}[t]{0.24 \textwidth}
        \includegraphics[width=\textwidth]{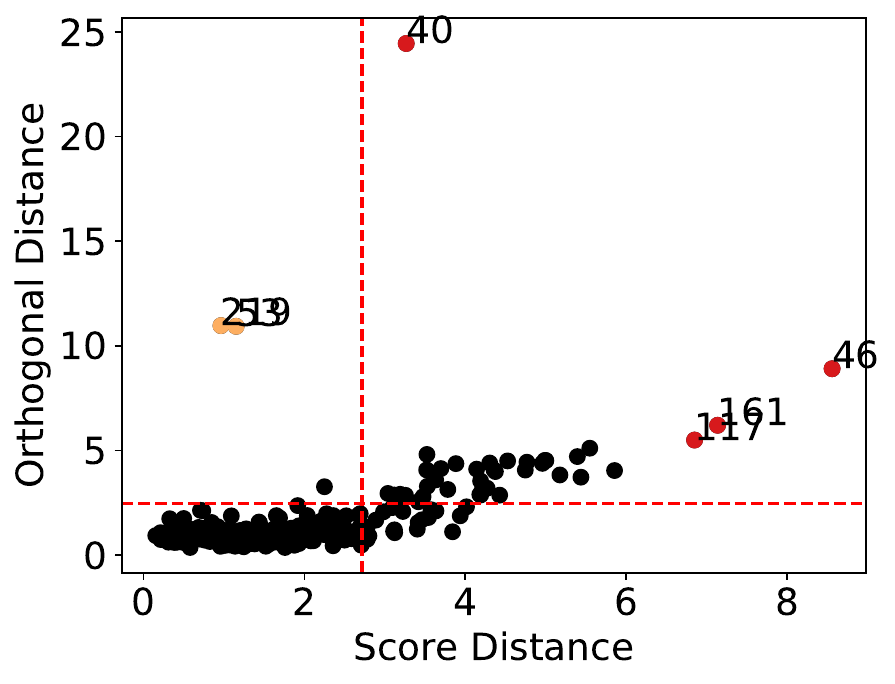}
        \caption{FIR-PCA}
        \label{subfig:topgear_FIR_PCA_outlier_map}
    \end{subfigure}
    \caption{Score (top row) and outlier maps (bottom row) for CPCA, DetMCD-PCA, FIR-PCA, and FIR-PCA. The red and orange points correspond to outliers considered in this example.}
    \label{fig:topgear}
\end{figure}

\subsection{Philips Data}
\label{subsec:TipGear}

The Philips dataset from ~\cite{Willems01012009} consists of $p=9$ characteristics of $n=677$ diaphragm parts for TV sets. We preprocessed the data by subtracting the mean from each feature column and normalizing it by its standard deviation. ~\cite{rousseeuw_fast_1999} showed that the data have clustered outliers and the points $491-566$ correspond to one of the clustered outliers. 

\autoref{fig:TVsets} shows the outlier map for the different methods with the outliers $491-566$ indicated in red. The execution time for CPCA, DetMCD-PCA, FDB-PCA, and FIRE are $10.90$ ms, $692.24$ ms, $35.22$ ms, and $43.96$ ms, respectively. FIR-PCA effectively separates the outliers(shown in red) and reveals three clusters found in ~\citep{Willems01012009} while the faster methods failed to separate the outliers (CPCA, FDB-PCA) and DetMCD-PCA only found two clusters. Overall, FIR-PCA projected the data faithful to its nature while only imposing a negligible amount of overhead compared to FDB-PCA, exhibiting great efficiency and robustness.
\begin{figure}[H]
    \centering
    \begin{subfigure}[t]{0.24 \textwidth}
        \includegraphics[width=\textwidth]{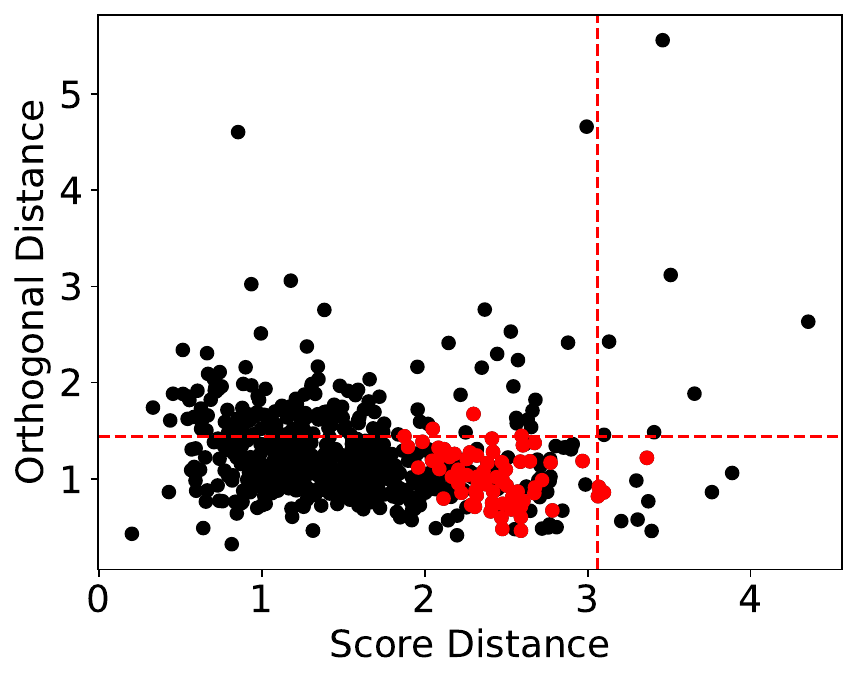}
        \caption{CPCA}
        \label{subfig:TVsets_C_PCA_outlier_map}
    \end{subfigure}
    \begin{subfigure}[t]{0.24 \textwidth}
        \includegraphics[width=\textwidth]{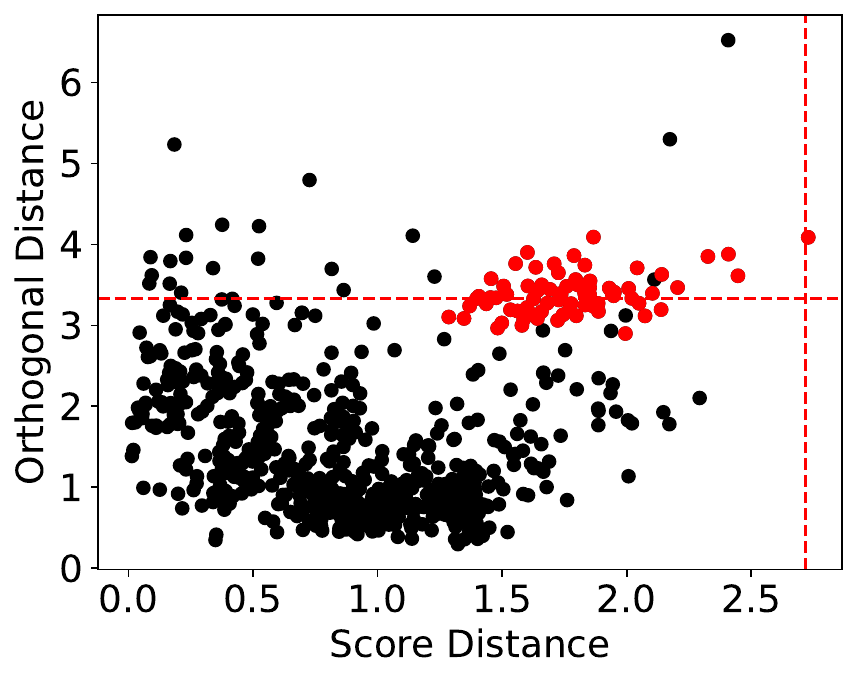}
        \caption{DetMCD-PCA}
        \label{subfig:TVsets_DetMCD_PCA_outlier_map}
    \end{subfigure}
    \begin{subfigure}[t]{0.24 \textwidth}
        \includegraphics[width=\textwidth]{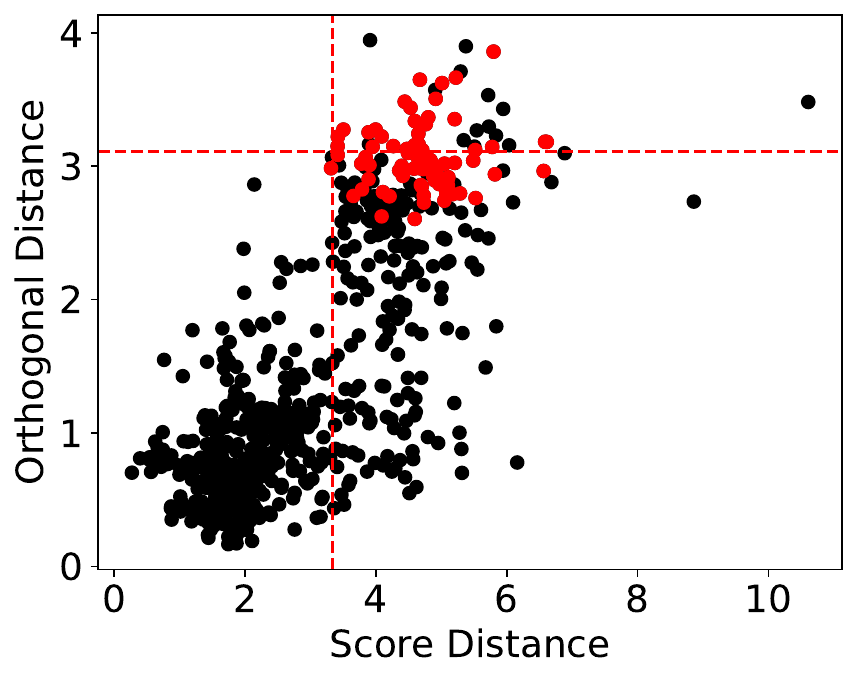}
        \caption{FDB-PCA}
        \label{subfig:TVsets_FDB_PCA_outlier_map}
    \end{subfigure}
    \begin{subfigure}[t]{0.24 \textwidth}
        \includegraphics[width=\textwidth]{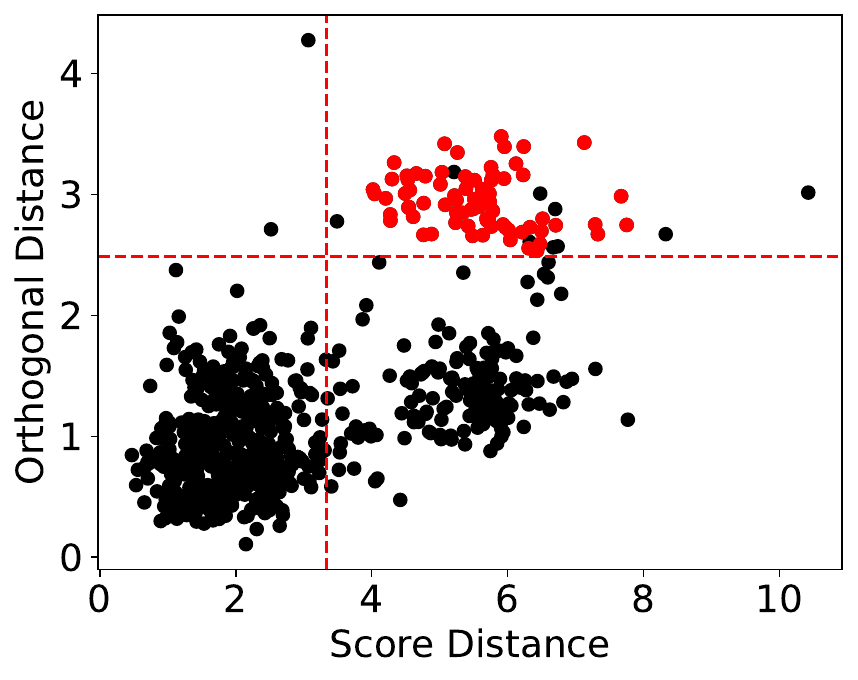}
        \caption{FIR-PCA}
        \label{subfig:TVsets_FIR_PCA_outlier_map}
    \end{subfigure}
    \caption{Outlier maps for CPCA, DetMCD-PCA, FIR-PCA, and FIR-PCA. The red points are previously identified outliers.}
    \label{fig:TVsets}
\end{figure}

\section{Conclusion}
\label{sec:conclusion}
We introduced FIR, a novel robust location and covariance estimation method, FIR-PCA, a robust PCA based on FIR, that efficiently handles outlier-contaminated data while maintaining computational efficiency. Our approach builds upon incremental PCA and employs a fast iterative selection mechanism to construct a robust estimate of location and covariance that is orthogonal equivariant and permutation invariant. Several examples demonstrate that the FIR method is more resilient to outliers than classical PCA and robust methods such as DetMCD and FDB. The FIR approach is not suitable for datasets where the number of features $p$ exceeds the number of observations $n$. In addition, the FIR method does not handle cases where corrupted outliers may contain partial inliers for some of the features. As this work continues, we plan to explore and develop techniques for handling cases where $p>n$, missing data, and partial outliers. In addition, we aim to extend our current method to higher-order tensors, enabling robust statistical analysis of more complex datasets beyond matrices.







\bibliographystyle{agsm}

\bibliography{references}
\end{document}